\documentclass[a4paper]{article}
\usepackage{amsmath,amsthm,amssymb, enumerate}
\usepackage{doi}
\usepackage{float}
\usepackage{nicefrac}
\usepackage{placeins}
\usepackage{pgfplots}
\pgfplotsset{compat=1.15}
\usepackage{mathrsfs}
\usetikzlibrary{arrows}
\usepackage{todonotes}
\usepackage[dvipsnames]{xcolor}
\usepackage[table]{xcolor}
\usepackage{complexity}
\usepackage{xspace}
\usepackage[capitalise]{cleveref}
\usepackage{subcaption}
\usepackage{authblk}
\usepackage{thm-restate}
\usepackage{fullpage}
\usepackage{tikz-cd}
\usepackage{extarrows}
\usepackage{centernot}
\usepackage{multirow}
\usepackage{array}
\usepackage{arydshln}
\usepackage{adjustbox}

\newcolumntype{C}[1]{>{\centering\arraybackslash}m{#1}}

\usetikzlibrary{decorations.markings,arrows.meta}

\newtheorem{theorem}{Theorem}
\newtheorem{corollary}[theorem]{Corollary}

\newtheorem{lemma}[theorem]{Lemma}

\newtheorem{conjecture}[theorem]{Conjecture}
\newtheorem{question}[theorem]{Question}
\newtheorem{remark}[theorem]{Remark}

\tikzset{               
        invisible/.style={opacity=0},
        visible on/.style={alt={#1{}{invisible}}},
        alt/.code args={<#1>#2#3}{%
          \alt<#1>{\pgfkeysalso{#2}}{\pgfkeysalso{#3}} 
        },
        cross/.style={cross out, thick,draw=black, minimum size=2*(#1-\pgflinewidth), inner sep=0pt, outer sep=0pt},
        cross/.default={0.25em},
        edge/.style={ultra thick,black},
        arc/.style={very thick,black,->},
        medge/.style={decorate,very thick,decoration={snake}},
        aedge/.style={very thick,dashed,black},
        dedge/.style={thick,->},
        availedge/.style={thick,blue},
        vertex/.style={inner sep=0.25em,shape=circle,thick,draw,node distance=4em},
        smalledge/.style={thick,DodgerBlue},
        smallvertex/.style={inner sep=0.1em,shape=circle,thick,draw=Black,node distance=4em}
}

\tikzset{
    midarrow/.style={
        postaction={
            decorate,
            decoration={
                markings,
                mark=at position 0.55 with {
                    \arrow{Latex[length=2mm]}
                }
            }
        }
    }
}

\title{On multiplicativity of directed graphs} 
\date{}

\author[1]{Soura Sena Das}
\author[2]{Moritz M\"uhlenthaler}
\author[3]{Sagnik Sen}
\author[2]{Thomas Suzan}

\affil[1]{Indian Statistical Institution, Kolkata, India}
\affil[2]{G-SCOP, Universit\'{e} Grenoble-Alpes, France}

\affil[3]{Indian Institute of Technology Dharwad, India}

\begin{document}

\maketitle

\begin{abstract}
	A graph category is a category with a set of graphs or similar structures (such as, directed graphs, signed graphs, etc.) playing the role of objects, and an  appropriate notion of homomorphism playing the role of morphisms.  The characterization of multiplicative objects are important open problems in categories of undirected and directed graphs. While the recent disproving of the Hedetniemi's conjecture due to Shitov (Ann. Math. 2019), which claimed that all complete graphs are multiplicative, provided a breakthrough in the study of multiplicative undirected graphs, the characterization of multiplicative undirected graphs remains known only for cycles, circular cliques $K_{{n/k}}$ where ${n/k} \in (2,4]$, complete graphs, and graphs whose each edge is part of at most one $4$-cycle.  Similarly, whether a given directed graph is multiplicative or not is known only for some oriented paths, oriented cycles, and transitive tournaments. 

    We study multiplicative graphs in the category of directed graphs where pushable homomorphism plays the role of morphism. We provide full multiplicativity characterization for directed bipartite graphs, oriented cycles, and transitive tournaments. As a consequence we find new (infinite) classes of non-multiplicative directed graphs in the usual directed graphs category. We also resolve an open question posed by Das \textit{et al.} (CALDAM 2026) related to the existence of exponential directed graphs with respect to pushable homomorphisms, and use our solution as a tool for our proofs. 
    
\end{abstract}

\section{Introduction}
In the 1960s, Hedetniemi famously conjectured~\cite{hedetniemi1966} that $\chi(G \times H) = \min\{\chi(G),\chi(H)\}$, where $G$ and $H$ are simple graphs, $\chi(\cdot)$ denotes the chromatic number, and $G \times H$ is the categorical product of $G$ and $H$. The structure of $G \times H$ is known as a graph on the vertex set $V(G) \times V(H)$ where $(u, v)$ is adjacent to $(u',v')$ if $u, u'$ and $v,v'$ are adjacent in $G$ and $H$, respectively. 
The conjecture remained open for many years until it was famously disproved by Shitov~\cite{shitov-counter} in 2019 using exponential graphs  as a tool. 
While the above formulation of the 
Hedetniemi's conjecture using the
notion of chromatic number is its most popular version, 
its equivalent formulation using the language of graph 
homomorphisms and multiplicative graphs possibly 
carry more insight and naturally leads to the big 
picture of characterizing multiplicative graphs in the 
category of undirected graphs. 
Let us try to build our 
motivation from that perspective.

Let $G$ and $H$ be two graphs. 
A \textit{homomorphism} of $G$ to $H$ is a vertex  mapping
$f : V(G) \to V(H)$ such that $f(u)f(v)$ is an edge  of $H$ whenever $uv$ is an edge  of $G$.
Additionally, if a homomorphsim $f$ preserves non-adjacencies, that is, 
$f(u),f(v)$ are non-adjacent in $H$ 
whenever $u, v$ are non-adjacent in $G$, then $f$ is an \textit{isomorphism}.  
We write $G \rightarrow H$ to indicate that $G$ admits a homomorphism to $H$. 
If $G \rightarrow H$ and $H \rightarrow G$, then $G$ and $H$ are \textit{homomorphically equivalent}.
A graph  $K$ is \textit{multiplicative} if $G \times H \rightarrow K$ implies $G \rightarrow K$ or $H \rightarrow K$. 
A graph $K$ is \textit{non-multiplicative} if 
it is not multiplicative. 
It is known (and is easy to observe) that 
 $\chi(G) \leq n$ if and only if $G \rightarrow K_n$, where $K_n$ denotes the complete graph on $n$ vertices. Therefore, the Hedetniemi's conjecture translates to the following statement. 

\begin{conjecture}[Hedetniemi's conjecture 1966~\cite{hedetniemi1966}] 
For all $n \geq 1$, the complete graph $K_n$ is multiplicative.
\end{conjecture}

In 1985, El-Zahar and Sauer~\cite{el-zahar-sauer} proved that $K_1$, $K_2$, and $K_3$ are multiplicative, marking the first significant progress on the conjecture. 
More recently in 2019, Shitov~\cite{shitov-counter} disproved Hedetniemi’s conjecture for  large values of $n$. 
Subsequent works~\cite{tardif2022chromatic,tardiff-kn4,wrochna2020counterexample,zhu2020note} eventually established that $K_n$ is non-multiplicative if $n \geq 4$. 
That means, at present we have a complete characterization of complete graphs with respect to multiplicativity. 

\begin{theorem}[El-Zahar and Sauer 1985~\cite{el-zahar-sauer}, Tardif 2023~\cite{tardiff-kn4}]\label{th kn-multiplicative}
    A complete graph $K_n$ is multiplicative if and only if $n \leq 3$. 
\end{theorem}

Thus, the above theorem provides a true and complete resolution of the Hedetniemi's conjecture by appropriately rephrasing its statement. A more general and a natural question that can be asked in this context is the following. 

\begin{question}\label{ques multiplicative}
     Can you characterize all multiplicative graphs? 
\end{question}
 
Most research related to the Hedetniemi's conjecture
in between the works of 
El-Zahar and Sauer~\cite{el-zahar-sauer}, and 
Shitov~\cite{shitov-counter} focused on the above question. Primarily, researchers tried to characterize multiplicative graphs in a particular graph family. 
As a positive result, and a natural extension to 
El-Zahar and Sauer~\cite{el-zahar-sauer}'s proof of 
$K_3$ is multiplicative, the following was proved.

\begin{theorem}[Hell, Zhou, and Zhu 1994~\cite{hell_oriented_multiplicativity}]\label{th simple cycles}
    All cycles are multiplicative. 
\end{theorem}

A \textit{circular clique} $K_{{n/k}}$ is a graph 
with vertex set 
$\mathbb{Z}/n\mathbb{Z} = \{0, 1, \cdots, n-1\}$ where 
two vertices $i$ and $j$ are adjacent if 
$|i-j| \geq k~(\bmod~n)$. 
Observe  that, $K_{2k+1/k}$ is 
nothing but the odd cycle on $2k+1$ vertices. As a generalization of Theorem~\ref{th simple cycles}, Tardif~\cite{tardiff-circ-cliq} proved the following. 

\begin{theorem}[Tardif 2005~\cite{tardiff-circ-cliq}]\label{th circular clique}
    The circular clique $K_{{n/k}}$ is multiplicative for all
    ${n/k} \in [2, 4)$. 
\end{theorem} 

Zhu~\cite{zhu1992star} had proposed a strengthening 
of the Hedetniemi's conjecture (now disproved) by replacing chromatic number 
with circular chromatic number in its popular formulation. The conjecture by Zhu~\cite{zhu1992star} is equivalent to saying all circular cliques $K_{{n/k}}$s are multiplicative. 
The above theorem was a preliminary attempt to prove the conjecture.
Since $K_{n/1} = K_n$, the conjecture by Zhu~\cite{zhu1992star}, if true, would have implied the 
Hedetniemi's conjecture as a special case. However, it is obviously false since the
Hedetniemi's conjecture has been disproved.  Still it will be interesting to characterize the circular cliques according to their multiplicativity.

The last result which improves our understanding of multiplicativity focuses on sparse graphs. 

\begin{theorem}[Tardif and Wrochna 2019~\cite{tardif-wrochna2019hedetniemi}]\label{th sparse}
    If $K$ is a graph where each of it's edge is part of at most one $4$-cycle, then $K$ is multiplicative. 
\end{theorem}

Notice that, even though Theorems~\ref{th kn-multiplicative}, \ref{th simple cycles}, \ref{th circular clique}, and~\ref{th sparse} helped us characterize many graphs with respect to 
multiplicativity, we are still very far from a complete answer to Question~\ref{ques multiplicative}. 

\medskip

In this article, we address the analogue of Question~\ref{ques multiplicative} for directed graphs with respect to the ordinary homomorphism, and also a recent variant called the pushable homomorphism. However, to express our contributions and their proofs, we need a bit of preliminaries including some basics 
of category theory. Therefore, we would like to present the organization of the article here, and defer further motivation and context from the study of directed graphs to the later sections.

\medskip

\noindent \textbf{Organization}
\begin{itemize}
    \item In Section~\ref{sec preliminaries of category theory}, we present the preliminaries of category theory. 

    \item In Section~\ref{sec multiplicativity of directed graphs}, we present 
    an overview of the multiplicativity research on directed graphs with respect to the ordinary homomorphism. We also discuss and present relevant known results as motivation. 

    \item In Section~\ref{sec pushable homomorphisms}, we introduce pushable homomorphism and provide a brief historical overview on its research. We also recall and  prove some basic results which, on the one hand builds the foundation for the rest of the article, and on the other hand 
    establishes a deep connection with the study of directed graphs with respect to the ordinary homomorphism. 

    \item In Section~\ref{sec pushable exponential}, we prove the existence of exponential objects in the category of pushable directed graphs and also use it as a tool for our proofs. This also positively settles an open question posed in~\cite{das2026update}.

    \item  In Section~\ref{sec pushable multiplicativity}, we completely characterize the multiplicativity of the following families of directed graphs
    with respect to pushable homomorphisms: bipartite  directed graphs, oriented cycles, transitive tournaments.

\item In Section~\ref{sec directed non-multiplicative}, 
using the connections between ordinary homomorphisms and pushable homomorphisms established in Section~\ref{sec pushable homomorphisms} and the pushable multiplicativity related results proved in Section~\ref{sec pushable multiplicativity}, we present new infinite families of non-multiplicative directed graphs.


    \item In Section~\ref{sec conclusions}, we summarize the state of the art in this domain, present some open questions,  and suggest future directions to conclude the article. 
\end{itemize}

\section{Preliminaries: Category Theory}\label{sec preliminaries of category theory}
A \textit{category} $\mathbf{C}$ is a mathematical system 
consisting of 
\begin{itemize}
    \item a class of objects denoted by $ob(\mathbf{C})$,

    \item for each pair of $X , Y \in ob(\mathbf{C})$, 
   a set of all  \textit{morphisms} of $X$ to $Y$ denoted by ${Hom}_{\mathbf{C}}(X,Y)$,

   \item for each triple of $X, Y, Z \in ob(\mathbf{C})$,
    an associative function  
   $$\circ: Hom_{\mathbf{C}}(X,Y) \times Hom_{\mathbf{C}}(Y,Z) \rightarrow Hom_{\mathbf{C}}(X,Z)$$ 
   called \textit{composition},

   \item for each $X \in ob(\mathbf{C})$, 
   an \textit{identity morphism} $id_X \in Hom_{\mathbf{C}}(X,X)$ which satisfies 
   $\circ(id_X , f) = f$ 
   and $\circ(g, id_X) = g$
   for any $f \in Hom_{\mathbf{C}}(Y,X)$, 
   $g \in Hom_{\mathbf{C}}(X,Z)$, and  $Y, Z \in ob(\mathbf{C})$.   
\end{itemize}
Usually, we will use the notation 
$f: X \to Y$ instead of $f \in Hom_{\mathbf{C}}(X,Y)$ to denote a morphism. In fact, the notation $X \to Y$ will represent the fact that there exists a morphism of $X$ to $Y$, that is, $Hom_{\mathbf{C}}(X,Y) \neq \emptyset$. Moreover, we will use the notation $f \circ g$ instead of $\circ(f,g)$ for the composition  
function. 

 \medskip

  The \textit{categorical product} of $X, Y \in ob(\mathbf{C})$  is an object $X \times Y \in ob(\mathbf{C})$ along with two morphisms 
  $p_X: X \times Y \to X$ 
  and 
  $p_Y \in X \times Y \to Y$
  called \textit{projections}
  satisfying the following: 
  for any $Z \in ob(\mathbf{C})$ with  
  $f_X: Z \to X$ 
  and 
  $f_Y: Z  \to Y$
  there exists a unique morphism 
  $g: Z \to X \times Y$
  such that 
  $f_X = p_X \circ g$ and $f_Y = p_Y \circ g$. The definition is captured by the following figure with a mandate of the diagram commuting. 

\begin{figure}[h]
    \centering
     \tikzset{
    mid-isharrow/.style={
        postaction={
            decorate,
            decoration={
                markings,
                mark=at position 0.7 with {
                    \arrow{Latex[length=2mm]}
                }
            }
        }
    }
}

\tikzset{
    endarrow/.style={
        postaction={
            decorate,
            decoration={
                markings,
                mark=at position 1 with {
                    \arrow{Latex[length=2mm]}
                }
            }
        }
    }
}

\tikzset{
  every label/.style={fill=white,inner sep=1pt}
}

\usetikzlibrary{backgrounds}

\pgfdeclarelayer{bg}
\pgfdeclarelayer{fg}
\pgfsetlayers{bg,main,fg}

\begin{tikzpicture}[scale=1,
vertex/.style={circle,fill=black,inner sep=1.5pt},
every label/.style={font=\small, fill=white!50, inner sep=1pt},
label distance=4pt]

    \begin{scope}[xshift=0cm]
        
        \node (x1) at (0,0) {$X$};
        \node (x1x2) at (2,0) {$X\times Y$};
        \node (x2) at (4,0) {$Y$};
        \node (y) at (2,1.5) {$Z$};

        \draw[endarrow] (x1x2)--node[left,below]{$p_X$}(x1);
        \draw[endarrow] (x1x2)--node[left,below]{$p_Y$}(x2);
        \draw[endarrow] (y)--node[left,yshift=0.5em]{$f_X$}(x1);
        \draw[endarrow,dashed] (y)--node[left]{$g$}(x1x2);
        \draw[endarrow] (y)--node[right,yshift=0.5em]{$f_Y$}(x2);

    \end{scope}

\end{tikzpicture}
    \end{figure}


The \textit{exponential object} of 
$X \in ob(\mathbf{C})$ to the 
$Y \in ob(\mathbf{C})$ is an object 
$X^Y$ along with a morphism 
$\epsilon: X^Y \times Y \to X$ called \textit{evaluation} satisfying the following:
for any $Z \in ob(\mathbf{C})$ and any morphism 
$g: Z \times Y \to X$ there exists a unique morphism $\lambda_g : Z \to X^Y$ such that 
$g = \epsilon \circ (\lambda_g \times 1_Y)$, 
where $\lambda_g \times 1_Y: Z \times Y \to X^Y \times Y$ is given by $(\lambda_g \times 1_Y)(z,y) = (\lambda_g(z), y)$ for all $ y \in Y,z \in Z$. 
The definition is captured by the following figure with a mandate of the diagram commuting. 

\begin{figure}[h]
    \centering
    \tikzset{
    mid-isharrow/.style={
        postaction={
            decorate,
            decoration={
                markings,
                mark=at position 0.7 with {
                    \arrow{Latex[length=2mm]}
                }
            }
        }
    }
}

\tikzset{
    endarrow/.style={
        postaction={
            decorate,
            decoration={
                markings,
                mark=at position 1 with {
                    \arrow{Latex[length=2mm]}
                }
            }
        }
    }
}

\tikzset{
  every label/.style={fill=white,inner sep=1pt}
}

\usetikzlibrary{backgrounds}

\pgfdeclarelayer{bg}
\pgfdeclarelayer{fg}
\pgfsetlayers{bg,main,fg}

\begin{tikzpicture}[scale=1,
vertex/.style={circle,fill=black,inner sep=1.5pt},
every label/.style={font=\small, fill=white!50, inner sep=1pt},
label distance=4pt]

    \begin{scope}[xshift=0cm]
        
        \node (zy) at (0,0) {$X^Y$};
        \node (x) at (0,2) {$Z$};

        \draw[endarrow,dashed] (x)--node[left]{$\lambda_g$}(zy);     
    
    \end{scope}

    \begin{scope}[xshift=3cm]

        \node (zy) at (0,0) {$X^Y\times Y$};
        \node (xy) at (0,2) {$Z\times Y$};
        \node (z) at (3,0) {$X$};

        \draw[endarrow,dashed] (xy)--node[left]{$\lambda_g\times 1_Y$}(zy);
        \draw[endarrow] (xy)--node[above]{$g$}(z);
        \draw[endarrow] (zy)--node[left,below]{$\epsilon$}(z);
        
    \end{scope}

\end{tikzpicture}
\end{figure}


\begin{remark}
    The set of all graphs  together with the notion of 
homomorphism forms the category of undirected graphs, denoted by \textbf{Und}. The earlier defined 
$G \times H$ actually satisfies the 
universal properties 
 of categorical products given in the above definition (see~\cite{hell-nesetril-book} for a proof). 
Thus, $G \times H$ is appropriately called  the categorical product of $G$ and $H$.
\end{remark}

\section{Multiplicativity of directed graphs}\label{sec multiplicativity of directed graphs}
A \textit{directed graph}  $\overrightarrow{G}$ consists of a vertex set $V(\overrightarrow{G})$ and an arc set $A(\overrightarrow{G})$, where each arc is an ordered pair of vertices. An arc $uv$ is directed from $u$ to $v$, and we write $uv \in A(\overrightarrow{G})$.
A \textit{bidirected arc} refers two arcs $uv$ and $vu$ in opposite directions between two vertices $u$ and $v$.
For a vertex $v \in V(\overrightarrow{G})$, the \textit{out-neighborhood} of $v$ is 
$N^+(v)=\{u \mid vu \in A(\overrightarrow{G})\}$, and the \textit{in-neighborhood} of $v$ is 
$N^-(v)=\{u \mid uv \in A(\overrightarrow{G})\}$. The \textit{out-degree} and \textit{in-degree} of $v$ are defined by $d^+(v)=|N^+(v)|$ and $d^-(v)=|N^-(v)|$, respectively.
An \textit{oriented graph} is a directed graph obtained from a simple undirected graph by assigning a direction to each edge; equivalently, it is a directed graph containing no loops or bidirected arcs.

Let $\overrightarrow{G}$ and $\overrightarrow{H}$ be directed graphs. 
A \textit{homomorphism} of $\overrightarrow{G}$
to $\overrightarrow{H}$ is a vertex mapping $f: V(\overrightarrow{G}) \rightarrow V(\overrightarrow{H})$
such that $f(u)f(v)$ is an arc  of $\overrightarrow{H}$ whenever $uv$ is an arc  of $\overrightarrow{G}$. 
Additionally, if a homomorphsim $f$ preserves non-adjacencies, that is, 
$f(u),f(v)$ are non-adjacent in $\overrightarrow{H}$ 
whenever $u, v$ are non-adjacent in $\overrightarrow{G}$, then $f$ is an \textit{isomorphism}.  
We write $\overrightarrow{G} \rightarrow 
\overrightarrow{H}$ to indicate that 
$\overrightarrow{G}$ admits a homomorphism 
to $\overrightarrow{H}$.
If $\overrightarrow{G} \rightarrow \overrightarrow{H}$ and $\overrightarrow{H} \rightarrow \overrightarrow{G}$, then $\overrightarrow{G}$ and $\overrightarrow{H}$ are \textit{homomorphically equivalent}.
The set of all directed graphs endowed with 
the above defined homomorphism forms the 
\textit{category of directed graphs}, denoted by \textit{\textbf{Dir}}.

It is known~\cite{hell-nesetril-book} that  the category of directed graphs contains the \textit{categorical product} of any two directed graphs 
$\overrightarrow{G}$ and $\overrightarrow{H}$,  
denoted by $\overrightarrow{G}\times \overrightarrow{H}$, 
whose structural description can be given as follows: 
 $\overrightarrow{G}\times \overrightarrow{H}$, 
is the directed graph  with vertex set $V(\overrightarrow{G}) \times V(\overrightarrow{H})$ 
in which there is an arc from the vertex $(u,v)$  to the vertex $(u',v')$ if $uu'$ is an arc in $\overrightarrow{G}$ and $vv'$ is an arc  in $\overrightarrow{H}$. 
Naturally, a directed graph $\overrightarrow{K}$ is  \textit{multiplicative} if $\overrightarrow{G} \times \overrightarrow{H} \rightarrow \overrightarrow{K}$ implies $\overrightarrow{G} \rightarrow \overrightarrow{K}$ or $\overrightarrow{H} \rightarrow \overrightarrow{K}$. 
A directed graph $\overrightarrow{K}$ is \textit{non-multiplicative} if 
it is not multiplicative. 

Since categorical product exists for directed graphs, 
the analogue of 
Question~\ref{ques multiplicative} becomes relevant.

\begin{question}\label{ques digraph multiplicative}
     Can you characterize all multiplicative directed graphs? 
\end{question}

There has been some progress in identifying 
multiplicative and non-multiplicative 
directed graphs which we are going to 
summarize in the following. Notice that, since all bipartite graphs (with at least one edge) are 
homomorphically equivalent to $K_2$, they are multiplicative by Theorem~\ref{th kn-multiplicative}. In contrast, the oriented paths have infinitely many homomorphically equivalent classes, and therefore, their multiplicativity characterization requires more attention.  The following is what is known about them.  

\begin{theorem}[Ne\v{s}et\v{r}il and Pultr 1978~\cite{nevsetvril-pultr}]
    Any oriented path that is homomorphically equivalent to a directed path is multiplicative. 
\end{theorem}

The next natural class to investigate are oriented cycles. A complete characterization 
with respect to multiplicativity is presented in Theorem~\ref{th oriented cycle}. We need some prerequisites to state it. 
Zhou~\cite{zhou92} introduced a special class of oriented cycles, called $\mathcal{C}$-cycle.
Let $\overrightarrow{B_n},\overrightarrow{S_n}$ and $\overrightarrow{T_n}$ be digraphs as depicted in Figure~\ref{fig c-cycles}. The class of $\mathcal{C}$-cycles is inductively defined as follows:
\begin{enumerate}
    \item Each $\overrightarrow{B_n}$ is a $\mathcal{C}$-cycle.

    \item Let $C$ be a $\mathcal{C}$-cycle and let $v$ be a vertex of out-degree $2$ (resp., in-degree $2$). Then there are two maximal directed paths $P,P'$ starting (resp., ending) at $v$, say of lengths $l\leq l'$. Let $m\leq l$ be an integer.
    Replace $v$ by $\overrightarrow{S_m}$ (resp., $\overrightarrow{T_m}$), identifying $a$ with  the beginning of $P$ and $b$ with the beginning of $P'$. The new digraph $C'$ is also a $\mathcal{C}$-cycle.

    \item There are no other $\mathcal{C}$-cycles.
\end{enumerate}

\begin{figure}
    \centering
    \usetikzlibrary{calc}
\tikzset{
    endarrow/.style={
        postaction={
            decorate,
            decoration={
                markings,
                mark=at position 1 with {
                    \arrow{Latex[length=2mm]}
                }
            }
        }
    }
}
\begin{tikzpicture}[scale=0.9,
vertex/.style={circle,fill=black,inner sep=1.5pt},
every label/.style={font=\small}]

\begin{scope}[xshift=0cm]

    \node[vertex,label=below:$a$] (a) at (0,0) {};
    \node[vertex,label=below:$b$] (b) at (3,0) {};
    \node[vertex,label=above:$u_n$] (un) at (1.5,2.5) {};
    \node[vertex,label=left:$u_1$] (u1) at (0.5,0.84) {};
    \node[vertex,label=left:$u_{n-1}$] (un-1) at (1,1.67) {};
    \node[vertex,label=right:$v_1$] (v1) at (2.5,0.84) {};
    \node[vertex,label=right:$v_{n-1}$] (vn-1) at (2,1.67) {};
    
    \foreach \t in {0.25,0.5,0.75}
    {
        \fill ($(u1)!\t!(un-1)$) circle (0.6pt);
    }
    \foreach \t in {0.25,0.5,0.75}
    {
        \fill ($(v1)!\t!(vn-1)$) circle (0.6pt);
    }
    
    \draw[endarrow] (b)--(a);
    \draw[endarrow] (a)--(u1);
    \draw[endarrow] (un-1)--(un);
    \draw[endarrow] (b)--(v1);
    \draw[endarrow] (vn-1)--(un);

    \node at (1.5,-1) {$\overrightarrow{B_n}$};

\end{scope}

\begin{scope}[xshift=4cm]

    \node[vertex,label=below:$a$] (a) at (0,0) {};
    \node[vertex,label=below:$b$] (b) at (3,0) {};
    \node[vertex,label=above:$u_n$] (un) at (1.5,2.5) {};
    \node[vertex,label=left:$u_1$] (u1) at (0.5,0.84) {};
    \node[vertex,label=left:$u_{n-1}$] (un-1) at (1,1.67) {};
    \node[vertex,label=right:$v_1$] (v1) at (2.5,0.84) {};
    \node[vertex,label=right:$v_{n-1}$] (vn-1) at (2,1.67) {};
    
    \foreach \t in {0.25,0.5,0.75}
    {
        \fill ($(u1)!\t!(un-1)$) circle (0.6pt);
    }
    \foreach \t in {0.25,0.5,0.75}
    {
        \fill ($(v1)!\t!(vn-1)$) circle (0.6pt);
    }
    
    \draw[endarrow] (a)--(u1);
    \draw[endarrow] (un-1)--(un);
    \draw[endarrow] (b)--(v1);
    \draw[endarrow] (vn-1)--(un);
    
    \node at (1.5,-1) {$\overrightarrow{S_n}$};

\end{scope}

\begin{scope}[xshift=8cm]

\node[vertex,label=above:$a$] (a) at (0,2.5) {};
\node[vertex,label=above:$b$] (b) at (3,2.5) {};
\node[vertex,label=below:$u_n$] (un) at (1.5,0) {};

\node[vertex,label=left:$u_1$] (u1) at (0.5,1.66) {};
\node[vertex,label=left:$u_{n-1}$] (unm) at (1,0.83) {};

\node[vertex,label=right:$v_1$] (v1) at (2.5,1.66) {};
\node[vertex,label=right:$v_{n-1}$] (vnm) at (2,0.83) {};

\foreach \t in {0.25,0.5,0.75}
{
    \fill ($(u1)!\t!(unm)$) circle (0.6pt);
}
\foreach \t in {0.25,0.5,0.75}
{
    \fill ($(v1)!\t!(vnm)$) circle (0.6pt);
}

\draw[endarrow] (u1)--(a);
\draw[endarrow] (un)--(unm);
\draw[endarrow] (v1)--(b);
\draw[endarrow] (un)--(vnm);

\node at (1.5,-1) {$\overrightarrow{T_n}$};

\end{scope}

\end{tikzpicture}
    \caption{The digraphs $\overrightarrow{B_n},\overrightarrow{S_n}$ and $\overrightarrow{T_n}$.}
    \label{fig c-cycles}
\end{figure}

\begin{theorem}[Hell, Zhou and Zhu 1994~\cite{hell_oriented_multiplicativity}, H\"aggkvist, Hell, Miller and Neumann Lara 1988~\cite{haggkvist_multiplicative_1988}]\label{th oriented cycle}
    An oriented cycle $\overrightarrow{C}$ is multiplicative if and only if it satisfies one of the three following conditions:
    \begin{enumerate}[(i)]

        \item $\overrightarrow{C}$ is homomorphically equivalent to a directed path,
    
        \item $\overrightarrow{C}$ is homomorphically equivalent to a directed cycle of a prime power length,

        \item $\overrightarrow{C}$ is a $\mathcal{C}$-cycle. 
    \end{enumerate}
\end{theorem}

While multiplicativity characterization of all tournaments (natural analogues of complete graphs) are not known, we do know that all transitive tournaments are multiplicative. 

\begin{theorem}[Ne\v{s}et\v{r}il and Pultr 1978~\cite{nevsetvril-pultr}]
    All transitive tournaments are multiplicative. 
\end{theorem}

An updated survey on multiplicativity in the categories of undirected and directed graphs is provided by Zhu~\cite{zhu2025survey}.
That means, the state of the art is far from finding the complete answer to Question~\ref{ques digraph multiplicative} as of now.
Thus, finding any new class of multiplicative or non-multiplicative directed graphs will be a worthwhile contribution in this research domain.

\section{Pushable homomorphisms: history and basic result}\label{sec pushable homomorphisms}
A particular modification via ``push'' operation on directed graphs has been popularly studied in the last few decades.
Introduced in 1972 by Mosesyan~\cite{mosesyan}, vertex pushing has been studied by Pretzel~\cite{pretzel1985,pretzel1986,pretzel1991}, Klostermeyer~\cite{klostermeyer1999pushing}, Babai and Cameron~\cite{babai2000} among others.
The concept was combined with homomorphisms  by Klostermeyer and MacGillivray~\cite{km2004}, who introduced  pushable homomorphisms and  pushable chromatic number. 
Since then, pushable homomorphisms and the pushable chromatic number have been further investigated: complexity dichotomy problems for pushable chromatic number were studied in~\cite{guegan_complexity_2015,km2004}, analogue of cliques were studied in~\cite{bensmail_oriented_2017}. 
Moreover, pushable chromatic number of various graph classes were studied, such as planar graphs and planar graphs with girth restrictions~\cite{borodin_universal_1998,das_pushable_2023,km2004,sen_homomorphisms_2017}, outerplanar graphs and outerplanar graphs with girth restrictions~\cite{km2004,sen_homomorphisms_2017}, graphs with bounded acyclic chromatic number~\cite{sen_homomorphisms_2017} and grids~\cite{bensmail_pushable_2023} among others.

To \textit{push} a vertex $v$ of a directed 
graph is to reverse the direction of the 
arcs incident to $v$. Let $\overrightarrow{G}$ and $\overrightarrow{G}'$
be two orientations
 of a graph $G$. If it is possible to obtain  $\overrightarrow{G}'$ from $\overrightarrow{G}$  through pushing a subset of vertices, then $\overrightarrow{G}$ is \textit{push equivalent} to $\overrightarrow{G}'$. 
  A \textit{pushable homomorphism} of $\overrightarrow{G}$
to $\overrightarrow{H}$ is a vertex mapping $f: V(\overrightarrow{G}) \rightarrow V(\overrightarrow{H})$
such that $f$ is a homomorphism of 
$\overrightarrow{G}'$ to $\overrightarrow{H}$, where $\overrightarrow{G}'$ is some directed graph push equivalent to $\overrightarrow{G}$. 
Additionally, if a pushable homomorphsim $f$ preserves non-adjacencies, that is, 
$f(u),f(v)$ are non-adjacent in $\overrightarrow{H}$ 
whenever $u, v$ are non-adjacent in $\overrightarrow{G}$, then $f$ is a \textit{pushable isomorphism}.  
We write $\overrightarrow{G} \xrightarrow{push} 
\overrightarrow{H}$ to indicate that 
$\overrightarrow{G}$ admits a pushable homomorphism to $\overrightarrow{H}$.
If $\overrightarrow{G} \xrightarrow{push} \overrightarrow{H}$ and $\overrightarrow{H} \xrightarrow{push} \overrightarrow{G}$, then $\overrightarrow{G}$ and $\overrightarrow{H}$ are \textit{pushably homomorphically equivalent}.
The set of all directed graphs endowed with 
the pushable homomorphism forms the 
\textit{category of pushable directed graphs}, denoted by \textit{\textbf{Push}}. 

Observe that, since pushable homomorphism 
allows modifications in a directed graph, 
it is not very clear whether a categorical product of two directed graphs always exists in the category of  pushable directed graphs 
or not. Recently~\cite{sen2026homomorphisms}, the existence of categorical products of directed graphs with respect to pushable homomorphism was proved 
as part of a generalized framework. Let us elaborate on this a bit.

A vertex $v$ of a directed graph $\overrightarrow{G}$ is \textit{push invariant} if it is not an isolated vertex, and yet, even after pushing $v$, $\overrightarrow{G}$ remains exactly the same (labeled) directed graph. Observe that a vertex $v$ can be push invariant if and only if its only adjacencies are through loop, or bidirected arcs.

\begin{figure}
    \centering
    \tikzset{
    mid-isharrow/.style={
        postaction={
            decorate,
            decoration={
                markings,
                mark=at position 0.7 with {
                    \arrow{Latex[length=2mm]}
                }
            }
        }
    }
}

\begin{tikzpicture}[scale=1,
vertex/.style={circle,fill=black,inner sep=1.5pt},
every label/.style={font=\small}]

\begin{scope}[xshift=0cm]

    \node[vertex,label=left:$u$] (u) at (0,0) {};
    \node[vertex,label=below:$v$] (v) at (2,0) {};
    \node[vertex,label=right:$w$] (w) at (4,0) {};
    
    \draw[midarrow,bend left=20] (u) to (v);
    
    \draw[midarrow,bend left=20] (v) to (u);

    \draw[midarrow] (w) -- (v);

\end{scope}

\begin{scope}[xshift=6cm]

    \node[vertex,label=left:$u$] (u) at (0,0) {};
    \node[vertex,label=below:$v$] (v) at (2,0) {};
    \node[vertex,label=right:$w$] (w) at (4,0) {};

    \node[vertex,label=above:$v*$] (v') at (2,2) {};
    \node[vertex,label=right:$w*$] (w') at (4,2) {};
    
    \draw[midarrow,bend left=20] (u) to (v);
    
    \draw[midarrow,bend left=20] (v) to (u);

    \draw[midarrow] (w) -- (v);
    
    \draw[midarrow,bend left=20] (v') to (u);
    \draw[mid-isharrow] (v')--(w);
    \draw[midarrow,bend left=20] (u) to (v');
    \draw[midarrow] (w')--(v');
    \draw[mid-isharrow] (v)--(w');

\end{scope}

\end{tikzpicture}
    \caption{A digraph $\overrightarrow{G}$ and its anti-twinned digraph $AT(\overrightarrow{G})$.}
    \label{fig anti-twin-eg}
\end{figure}

An \textit{anti-twin} of a vertex $v$ is another vertex $v^*$ satisfying $N^+(v^*) = N^-(v)$ and $N^-(v^*) = N^+(v)$.
The \textit{anti-twinned directed graph} of $\overrightarrow{G}$, denoted by $AT(\overrightarrow{G})$, is obtained by adding an anti-twin $v^*$ to each vertex $v$, that is not a push invariant vertex, to the graph $\overrightarrow{G}$. 
As a convention, we write $V(AT(\overrightarrow{G})) = V(\overrightarrow{G}) 
\cup V^*(\overrightarrow{G})$, where $V(\overrightarrow{G})$ is the set of original vertices from $\overrightarrow{G}$ in $AT(\overrightarrow{G})$, and $V^*(\overrightarrow{G})$ is  the set of newly added anti-twin vertices in $AT(\overrightarrow{G})$. 
This definition is a slight generalization of the existing notion of an anti-twinned directed graph, as the earlier versions were defined primarily on oriented graphs, and the push invariant vertices were just the isolated vertices. In the earlier definition, the condition of not creating an anti-twin of the push invariant vertices was not mentioned. However,  all the results and proofs can be trivially generalized, and holds  for this extended definition of the anti-twinned directed graph. We are going to recall some of the useful properties of the anti-twinned directed graphs, and also describe the construction of the categorical product 
of directed graphs with respect to pushable homomorphisms.

\begin{theorem}[\cite{km2004}]\label{th old hom}
    Let  $\overrightarrow{G}$ and $\overrightarrow{H}$ be directed graphs. Then the following are equivalent. 
    \begin{enumerate}[(i)]
        \item $\overrightarrow{G} \xrightarrow{push} \overrightarrow{H}$,

        \item $AT(\overrightarrow{G}) \rightarrow AT(\overrightarrow{H})$,

        \item $\overrightarrow{G} \rightarrow AT(\overrightarrow{H})$.
 
    \end{enumerate}
\end{theorem}

\begin{theorem}[\cite{sen_homomorphisms_2017}]\label{th AT isomorphic}
   Let  $\overrightarrow{G}$ and $\overrightarrow{H}$ be directed graphs. Then $AT(\overrightarrow{G})$ is isomorphic to $AT(\overrightarrow{H})$
        if and only if $\overrightarrow{G}$ is pushably isomorphic to $\overrightarrow{H}$.
\end{theorem}

Let us fix some conventions now. 
Given the anti-twinned graph $AT(\overrightarrow{G})$ 
of a directed graph $\overrightarrow{G}$, 
notice that for every vertex $v$ (not push invariant) of $\overrightarrow{G}$
 an anti-twin $v^*$ is added to $AT(\overrightarrow{G})$. This gives an one-to-one correspondence between the vertices of $\overrightarrow{G}$ that are not push invariant, and the newly added anti-twins in $AT(\overrightarrow{G})$. This one-to-one correspondence can be captured through setting up the convention 
 $(v^*)^* = v^{**} = v$. Moreover, if $v$ is a push invariant vertex, then set the convention
 $v^{**} = v^* = v$.

 Given two directed graphs $\overrightarrow{G}$ and $\overrightarrow{H}$, a vertex mapping 
 $f: V(AT(\overrightarrow{G})) \to V(AT(\overrightarrow{H}))$
is called an \textit{anti-twinned homomorphism} 
if $f$ is a homomorphism of 
$AT(\overrightarrow{G})$ to $AT(\overrightarrow{H})$ satisfying 
$f(v^*) = f(v)^*$ for all $v \in V(AT(\overrightarrow{G}))$. 
While not every homomorphism of $AT(\overrightarrow{G})$ to $AT(\overrightarrow{H})$ is an anti-twinned homomorphism, if $AT(\overrightarrow{G}) \rightarrow AT(\overrightarrow{H})$, then there exists an anti-twinned 
homomorphism of $AT(\overrightarrow{G})$ to $AT(\overrightarrow{H})$.

\begin{lemma}\label{lem hom to anti-twinned}
Suppose $f$ is a homomorphism of $AT(\overrightarrow{G})$ to $AT(\overrightarrow{H})$. Then there exists an anti-twinned homomorphism $f^*$
    of $AT(\overrightarrow{G})$ to $AT(\overrightarrow{H})$. 
\end{lemma}

\begin{proof} Let $f'$ be the function obtained by restricting the domain of $f$ to $V(\overrightarrow{G})$. Consider the following extension of $f'$:
    $$
f^*(v)=\begin{cases}
	f'(v), & \text{if $v\in V(\overrightarrow{G})$,}\\
    (f'(v^*))^*, & \text{if $v \in V^*(\overrightarrow{G}$)}.
		 \end{cases}
$$
We want to show that $f^*$ is an anti-twinned homomorphism of 
 $AT(\overrightarrow{G})$ to $AT(\overrightarrow{H})$.
That means, we need to verify the following two points:

\begin{itemize}
    \item for any vertex $v$ of $AT(\overrightarrow{G})$, we must have $f^*(v^*) = f^*(v)^*$,

    \item for any arc  $uv$ of $AT(\overrightarrow{G})$, $f^*(u)f^*(v)$ is an arc of $AT(\overrightarrow{H})$.
\end{itemize}

The first point is satisfied as per the definition of $f^*$. For verifying the second point, we need to consider the following three cases.

\begin{enumerate}[(i)]
    \item Suppose $uv$ is an arc of $AT(\overrightarrow{G})$  and both $u,v \in V(\overrightarrow{G})$. Then 
    $f^*(u)f^*(v)$ is an arc of $AT(\overrightarrow{H})$
    since $f^*(u)=f'(u)=f(u)$, 
        $f^*(v)=f'(v)=f(v)$ 
        and $f$ is a homomorphism.

    \item Suppose $uv$ is an arc of $AT(\overrightarrow{G})$ and exactly one of $u,v$ belongs to $V(\overrightarrow{G})$. Without loss of generality, let us assume that $u \in V(\overrightarrow{G})$ and  
    $v \in V^*(\overrightarrow{G})$. 
    Since $uv$ is an arc, note that, $v^*u$ is an arc
    of $AT(\overrightarrow{G})$. Moreover, observe that, 
    both $u$ and $v^*$ belongs to $V(\overrightarrow{G})$. 
    Thus,     $f^*(v^*)f^*(u)$ is an arc of $AT(\overrightarrow{H})$
    since $f^*(u)=f'(u)=f(u)$, 
        $f^*(v^*)=f'(v^*)=f(v^*)$ 
        and $f$ is a homomorphism.
    That means, $f^*(u)(f^*(v^*))^*$ is an arc of 
    $AT(\overrightarrow{H})$. 
    However, $(f^*(v^*))^* = (f'(v^*))^* = f^*(v)$, and thus 
    $f^*(u)f^*(v)$ is an arc of $AT(\overrightarrow{H})$.

    \item Suppose $uv$ is an arc of$AT(\overrightarrow{G})$  and both $u,v \in V^*(\overrightarrow{G})$. Then
    $u^*v^*$ is an arc of $AT(\overrightarrow{G})$
    while $u^*, v^* \in V(\overrightarrow{G})$. 
    Thus, 
    $f^*(u^*)f^*(v^*)$ is an arc of $AT(\overrightarrow{H})$
    since $f^*(u^*)=f'(u^*)=f(u^*)$, 
        $f^*(v^*)=f'(v^*)=f(v^*)$ 
        and $f$ is a homomorphism.
    Note that, $f^*(u^*)f^*(v^*)$ is an arc implies 
    $(f^*(u^*))^*(f^*(v^*))^*$ is an arc of $AT(\overrightarrow{H})$. 
    However, $(f^*(u^*))^* = (f'(u^*))^* = f^*(u)$ and 
    $(f^*(v^*))^* = (f'(v^*))^* = f^*(v)$. Therefore, 
    $f^*(u)f^*(v)$ is an arc of $AT(\overrightarrow{H})$. 
\end{enumerate}

This proves that $f^*$ is an anti-twinned homomorphism of $AT(\overrightarrow{G})$ to $AT(\overrightarrow{H})$.
\end{proof}

Let $f$ be 
 an \textit{anti-twinned homomorphism} 
 of 
$AT(\overrightarrow{G})$ to $AT(\overrightarrow{H})$.
The function 
$f_{res}: V(\overrightarrow{G}) \rightarrow V(\overrightarrow{H})$ given by 

$$
f_{res}(v)=\begin{cases}
			f(v), & \text{if $f(v) \in V(\overrightarrow{H})$,}\\
            f(v)^*, & \text{if $f(v) \in V^*(\overrightarrow{H})$}
		 \end{cases}
$$
for all $v \in V(\overrightarrow{G})$. 
This function is called an \textit{anti-twinned restriction} of $f$. Observe that, according to our definition, it is not clear whether $f_{res}$ is a pushable homomorphism or not. The following result ensures it. 

\begin{lemma}\label{lem anti-twinned to push}
    For every anti-twinned homomorphism $f$ of $AT(\overrightarrow{G})$ to $AT(\overrightarrow{H})$, its
    anti-twinned restriction $f_{res}$ is a pushable homomorphism of 
    $\overrightarrow{G}$ to $\overrightarrow{H}$. 
\end{lemma}

\begin{proof}
    Let 
    $S = \{v \in V(\overrightarrow{G}) | f(v) \in V^*(\overrightarrow{H})\}$. Let $\overrightarrow{G}'$ be the directed graph obtained from $\overrightarrow{G}$ by pushing the vertices of 
    $S$. 
    It is enough to show that $f_{res}$ is a homomorphism of $\overrightarrow{G}'$ to $\overrightarrow{H}$. We will verify this claim in three following cases. 
    
    \begin{enumerate}[(i)]
        \item Suppose  $uv$ is an arc of $\overrightarrow{G}$ 
        and $u, v \not\in S$. Then $uv$ is an arc of $\overrightarrow{G}'$ and 
        $f(u), f(v) \in V(\overrightarrow{G})$. 
        Since $f$ is a homomorphism, $f(u)f(v)$ is an arc of $AT(\overrightarrow{H})$. 
        Since, $f_{res}(u) = f(u)$ and $f_{res}(v) = f(v)$,
         we have $f_{res}(u)f_{res}(v)$ is an arc of $\overrightarrow{H}$.

\item Suppose  $uv$ is an arc of $\overrightarrow{G}$ 
        and exactly one of $u,v$ belongs to $S$. Without loss of generality, we may assume that $u \not\in S$ and 
        $v \in S$. Then $vu$ is an arc of $\overrightarrow{G}'$  (as only $v$  got pushed); 
        and $f(u) \in V(\overrightarrow{H})$ and 
        $f(v) \in V^*(\overrightarrow{H})$. 
        Since $f$ is a homomorphism, 
        $f(u)f(v)$ is an arc of $AT(\overrightarrow{H})$, which implies $f(v)^*f(u)$ is also an arc of 
        $AT(\overrightarrow{H})$. 
        Since by the definition of 
        $f_{res}$, we have $f_{res}(u) = f(u)$ and $f_{res}(v) = f(v)^*$,
         $f_{res}(v)f_{res}(u)$ is an arc of $\overrightarrow{H}$.

\item Suppose  $uv$ is an arc of $\overrightarrow{G}$ 
        and $u, v \in S$. Then $uv$ is an arc of $\overrightarrow{G}'$ as well (as both vertices got pushed) and $f(u), f(v)$ are vertices in $V^*(\overrightarrow{H})$. 
        Since $f$ is a homomorphism, 
        $f(u)f(v)$ is an arc of $AT(\overrightarrow{H})$, which implies $f(u)^*f(v)^*$ is also an arc of 
        $AT(\overrightarrow{H})$. 
        Since by the definition of 
        $f_{res}$, we have $f_{res}(u) = f(u)^*$ and $f_{res}(v) = f(v)^*$,
         $f_{res}(u)f_{res}(v)$ is an arc of $\overrightarrow{H}$. 
    \end{enumerate}

    Therefore $f_{res}$ is a homomorphism of $\overrightarrow{G}'$ to $\overrightarrow{H}$.
    By definition, this means that $f_{res}$ is a pushable homomorphism of $\overrightarrow{G}$ to $\overrightarrow{H}$.
\end{proof}

Given any pushable homomorphism $f$ of $\overrightarrow{G}$ to $\overrightarrow{H}$, suppose  $f_{ext}$ is an anti-twinned homomorphism of $AT(\overrightarrow{G})$ to $AT(\overrightarrow{H})$
whose anti-twinned restriction is $f$. Then $f_{ext}$ is called an \textit{anti-twinned extension} of $f$. The following result ensures that such an extension always exists. 

\begin{lemma}\label{lem push to anti-twinned extension}
    For every pushable homomorphism $f$ of 
    $\overrightarrow{G}$ to $\overrightarrow{H}$, there exists at least one anti-twinned extension $f_{ext}$, that is, an
    anti-twinned homomorphism $f_{ext}$ of $AT(\overrightarrow{G})$ to $AT(\overrightarrow{H})$ whose anti-twinned restriction is $f$. 
\end{lemma}

\begin{proof}
    Since $f$ is a pushable homomorphism of 
    $\overrightarrow{G}$ to $\overrightarrow{H}$, there must exist a set $S \subseteq V(\overrightarrow{G})$
    such that $f$ is a homomorphism of $\overrightarrow{G}'$
    to $\overrightarrow{H}$, where $\overrightarrow{G}'$
    is obtained from $\overrightarrow{G}$ by pushing the 
    vertices of $S$. 
    Define the function $f': V(\overrightarrow{G}) \to V(AT(\overrightarrow{H}))$ as follows
    $$
f'(v)=\begin{cases}
			f(v), & \text{if $v \not\in S$,}\\
            f(v)^*, & \text{if $v \in S$}.
		 \end{cases}
$$
Let us show that $f'$ is a homomorphism. 
Suppose that  $uv$ is an arc of $\overrightarrow{G}$ and consider the following two cases. 

\begin{enumerate}[(i)]
    \item If both $u,v \not\in S$ (resp., $u,v \in S$), 
then $uv$ is an arc of $\overrightarrow{G}'$ (since either none of them got pushed, or both of them got pushed) as well. 
Therefore, 
$f(u)f(v)$ is also an arc of $\overrightarrow{H}$, and $AT(\overrightarrow{H})$. This also means that 
$f(u)^*f(v)^*$ is an arc of $AT(\overrightarrow{H})$.
Since $f'(u) = f(u)$ and $f'(v) = f(v)$ 
(resp., $f'(u) = f(u)^*$ and $f'(v) = f(v)^*$), 
$f'(u)f'(v)$ is an arc of $AT(\overrightarrow{H})$.

\item If exactly one among $u,v$ belongs to $S$, then without loss of generality we may assume that 
$u \not\in S$ and $v \in S$. 
In this scenario, $vu$ is an arc in $\overrightarrow{G}'$, which implies that $f(v)f(u)$ is an arc of $\overrightarrow{H}$. That means, $f(u)f(v)^*$ is an arc of 
$AT(\overrightarrow{H})$. 
Since $f'(u) = f(u)$ and  $f'(v) = f(v)^*$, 
$f'(u)f'(v)$ is an arc of $AT(\overrightarrow{H})$.
\end{enumerate}
Therefore, $f'$ is a homomorphism.

Next, define the function $f_{ext}: AT(\overrightarrow{G}) \to AT(\overrightarrow{H})$ as follows:
   $$
f_{ext}(v)=\begin{cases}
			f'(v), & \text{if $v \in V(\overrightarrow{G})$,}\\
            f'(v^*)^*, & \text{if $v \in V^*(\overrightarrow{G})$}.
		 \end{cases}
$$
   Therefore, it is enough to prove that $f_{ext}$ is an anti-twinned extension of $f$.  

First let us show that $f_{ext}$ is a homomorphism of 
$AT(\overrightarrow{G})$ to $AT(\overrightarrow{H})$ via a verification of the following three cases. 

\begin{enumerate}[(i)]
    \item Suppose $uv$ is an arc of $AT(\overrightarrow{G})$ and both $u,v \in V(\overrightarrow{G})$. Then 
    $f_{ext}(u)f_{ext}(v)$ is an arc of $AT(\overrightarrow{H})$
    since $f_{ext}(u)=f'(u)$, 
        $f_{ext}(v)=f'(v)$ 
        and $f'$ is a homomorphism.

    \item Suppose $uv$ is an arc of $AT(\overrightarrow{G})$ and exactly one of $u,v$ belongs to $V(\overrightarrow{G})$. Without loss of generality let us assume that $u \in V(\overrightarrow{G})$ and  
    $v \in V^*(\overrightarrow{G})$. 
    Since $uv$ is an arc, note that, $v^*u$ is an arc 
    of $AT(\overrightarrow{G})$. Moreover, observe that, 
    both $u$ and $v^*$ belongs to $V(\overrightarrow{G})$. 
    Thus,     $f_{ext}(v^*)f_{ext}(u)$ is an arc of $AT(\overrightarrow{H})$
    since $f_{ext}(u)=f'(u)$, 
        $f_{ext}(v^*)=f'(v^*)$ 
        and $f'$ is a homomorphism.
    That means, $f_{ext}(u)f_{ext}(v^*)^*$ is an arc of 
    $AT(\overrightarrow{H})$. 
    However, $f_{ext}(v^*)^* = f'(v^*)^* = f_{ext}(v)$, and thus 
    $f_{ext}(u)f_{ext}(v)$ is an arc of $AT(\overrightarrow{H})$.

    \item Suppose $uv$ is an arc of $AT(\overrightarrow{G})$  and both $u,v \in V^*(\overrightarrow{G})$. Then
    $u^*v^*$ is an arc of $AT(\overrightarrow{G})$
    while $u^*, v^* \in V(\overrightarrow{G})$. 
    Thus, 
    $f'(u^*)f'(v^*)$ is an arc of $AT(\overrightarrow{H})$
    since $f'$ is a homomorphism. Hence,
    $f'(u^*)^*f'(v^*)^*$ is also an arc of $AT(\overrightarrow{H})$.
    However, $f_{ext}(u) = f'(u^*)^* $ and 
    $f_{ext}(v) = f'(v^*)^*$. Therefore, 
    $f_{ext}(u)f_{ext}(v)$ is an arc of $AT(\overrightarrow{H})$. 
\end{enumerate}
This proves that $f_{ext}$ is a homomorphism. 

Notice that, $f_{ext}(v^*) = (f_{ext}(v))^*$
by definition. Let us verify this. 
On the one hand, if $v \in V(\overrightarrow{G})$, then $$f_{ext}(v^*) = (f'(v^{**})^* = (f'(v))^* = (f_{ext}(v))^*.$$
On the other hand, 
if $v \in V^*(\overrightarrow{G})$, then 
$$f_{ext}(v^*) = f'(v^*) = (f'(v^*))^{**} = ((f'(v^*))^*)^* = (f_{ext}(v))^*.$$

The last part of the proof that remains is to check $f$ is an anti-twinned restriction of $f_{ext}$. Thus we need to show that, given any $v \in V(\overrightarrow{G})$, 
$$
f(v) = \begin{cases}
			f_{ext}(v), & \text{if $f_{ext}(v) \in V(\overrightarrow{H})$,}\\
            f_{ext}(v)^*, & \text{if $f_{ext}(v) \in V^*(\overrightarrow{H})$}.
		 \end{cases}
$$
Suppose $f_{ext}(v) \in V(\overrightarrow{H})$. 
Then $f_{ext}(v) = f'(v) \in   V(\overrightarrow{H})$. That means, $v \not\in S$, which implies $f_{ext}(v) = f'(v) = f(v)$. 
Next suppose $f_{ext}(v) \in V^*(\overrightarrow{H})$.
Then $f_{ext}(v) = f'(v) \in   V^*(\overrightarrow{H})$.
That means, $v \in S$, which implies $f_{ext}(v) = f'(v) = f(v)^*$. This further implies, $f_{ext}(v)^* = f(v)^{**} = f(v)$. 

This completes the proof. 
\end{proof}

Two anti-twinned homomorphisms $f_1, f_2$ of 
$AT(\overrightarrow{G})$ to $AT(\overrightarrow{H})$
are \textit{equivalent} if $f_1(v) = f_2(v)$ or $f_2(v^*)$ for all $v \in V(AT(\overrightarrow{G}))$. In other words, $f_1$ and $f_2$ has the same anti-twinned restrictions. If two anti-twinned homomorphisms $f_1, f_2$ are equivalent, then we denote it by $f_1 \sim  f_2$.  Observe that the above defined relation $\sim$ is an equivalence relation. 

\begin{lemma}\label{lem equiv anti-twinned extension}
    Let $f$ be a pushable homomorphism of 
    $\overrightarrow{G}$ to $\overrightarrow{H}$. Any two anti-twinned extensions of $f$ are equivalent. Moreover, any two equivalent anti-twinned homomorphisms have the same anti-twinned restrictions. 
\end{lemma}

\begin{proof}
    Let $f_1$ and $f_2$ be two anti-twinned extensions of $f$ and let 
    $v \in V(AT(\overrightarrow{G}))$. 
    Thus, according to the definitions of anti-twinned extensions, 
 $f_1(v), f_2(v) \in \{f(v), f(v)^*\}$ if 
    $v \in V(\overrightarrow{G})$ and 
    $f_1(v), f_2(v) \in \{f(v^*), f(v^*)^*\}$ if 
    $v \in V^*(\overrightarrow{G})$. 
    This implies, $f_1(v) = f_2(v)$ or $f_2(v^*)$, and hence they are equivalent. 

    Suppose $f_1$ and $f_2$ be two 
    equivalent anti-twinned homomorphisms
    of $AT(\overrightarrow{G})$ to $AT(\overrightarrow{H})$. 
    By definition, they have the same anti-twinned restriction.    
\end{proof}

\section{Pushable categorical products and multiplicativity}
Let $\overrightarrow{G}$
and $\overrightarrow{H}$ be two directed graphs. 
Let $V_{1}(\overrightarrow{G})$ denote the set of vertices  from $V(\overrightarrow{G})$ in $AT(\overrightarrow{G})$) such that they are push invariant in $\overrightarrow{G}$, 
and let $V_{2}(\overrightarrow{G}) = V(\overrightarrow{G}) \setminus V_{1}(\overrightarrow{G})$.
Similarly, let  $V_{1}(\overrightarrow{H})$ denote the set of vertices  from $V(\overrightarrow{H})$ in $AT(\overrightarrow{H})$) such that they are push invariant in $\overrightarrow{H}$, 
and let $V_{2}(\overrightarrow{H}) = V(\overrightarrow{H}) \setminus V_{1}(\overrightarrow{H})$.
Let $X = (V(\overrightarrow{G}) \times V(\overrightarrow{H})) \cup (V_2(\overrightarrow{G}) \times V^*(\overrightarrow{H}))$. 
The \textit{pushable categorical product} of $\overrightarrow{G}$
and $\overrightarrow{H}$, 
denoted by $\overrightarrow{G} \times_p \overrightarrow{H}$,  
is the following directed graph 
$$(AT(\overrightarrow{G}) \times AT(\overrightarrow{H}))[X].$$ 
That is, $\overrightarrow{G} \times_p \overrightarrow{H}$ is obtained by first taking the categorical product 
$AT(\overrightarrow{G}) \times AT(\overrightarrow{H})$ of the 
anti-twin directed graphs of $\overrightarrow{G}$ and $\overrightarrow{H}$, and furthermore its subgraph induced by 
the vertex subset $X$. 
To be the true categorical product in the category of pushable directed graphs, $\overrightarrow{G} \times_p \overrightarrow{H}$ must satisfy the 
universal properties of categorical product as per its definition. That is what we will show next.

\begin{theorem}\label{thm push-product}
     Let $\overrightarrow{G}$
and $\overrightarrow{H}$ be two directed graphs. Then the directed graph $\overrightarrow{G} \times_p \overrightarrow{H}$ is a categorical product in the category of pushable directed graphs. 
\end{theorem}

\begin{proof}
As an important first step we are going to show that 
$AT(\overrightarrow{G} \times_p \overrightarrow{H}) = AT(\overrightarrow{G}) \times AT(\overrightarrow{H})$. 
Since $AT(\overrightarrow{G} \times_p \overrightarrow{H})$
is an induced subgraph of $AT(\overrightarrow{G}) \times AT(\overrightarrow{H})$, to prove the above relation it is enough find appropriate anti-twins
for every vertex of 
$\overrightarrow{G} \times_p \overrightarrow{H}$ in 
$AT(\overrightarrow{G}) \times AT(\overrightarrow{H})$. 
Recall that, according to the definition of anti-twinned directed graphs, push invariant vertices are anti-twins of themselves.

Let $g$ and $h$ be any vertex of $\overrightarrow{G}$ and $\overrightarrow{H}$, respectively. 
Then the vertices of $\overrightarrow{G} \times_p \overrightarrow{H}$ are of the form $(g,h)$ where $g$ and $h$ varies over all vertices of their respective directed graphs, and of the form $(g,h^*)$ where $g$ and $h$ varies over all vertices of their respective directed graphs except the push invariant vertices. Let us find the anti-twinned vertices casewise. 
\begin{enumerate}[(i)]
    \item The vertex $(g,h)$ is push invariant if and only if $g$ and $h$ are push invariant. Therefore, if both $g$ and $h$ are push invariant, then $(g,h)^* = (g,h)$.

    \item If $g$ is push invariant, but $h$ is not, then 
    $(g,h)^* = (g,h^*)$. Similarly, 
    if $h$ is push invariant, but $g$ is not, then 
    $(g,h)^* = (g^*,h)$.

    \item If none among $g$ and $h$ are push invariant, then 
    $(g,h)^* = (g^*,h^*)$, and $(g,h^*)^* = (g^*,h)$. 
\end{enumerate}
Thus, every vertex of $\overrightarrow{G} \times_p \overrightarrow{H}$ has appropriate anti-twins in 
$AT(\overrightarrow{G}) \times AT(\overrightarrow{H})$ proving 
$AT(\overrightarrow{G} \times_p \overrightarrow{H}) = AT(\overrightarrow{G}) \times AT(\overrightarrow{H})$.

We know that $AT(\overrightarrow{G}) \times 
AT(\overrightarrow{H})$ is a categorical product of 
$AT(\overrightarrow{G})$ and $AT(\overrightarrow{H})$ where 
the projections are given by $p^G(g,h) = g$ and $p^H(g,h) = h$, respectively. Notice that, both the projections are natural anti-twinned homomorphisms. Let $p^G_{res}$ and 
$p^H_{res}$ be the anti-twin restrictions of $p^G$ and $p^H$. These are the projections of 
$\overrightarrow{G} \times_p \overrightarrow{H}$ to $\overrightarrow{G}$ and $\overrightarrow{H}$, respectively.

Suppose, $\overrightarrow{P}$ is a directed graph with 
pushable homomorphisms $f^G$ and $f^H$ of 
$\overrightarrow{P}$ to $\overrightarrow{G}$ and $\overrightarrow{H}$, respectively. Thus, we 
need to find a unique pushable homomorphism $\phi$ of  $\overrightarrow{P}$ to 
$\overrightarrow{G} \times_p \overrightarrow{H}$ which satisfies $p^G_{res} \circ \phi = f^G$ and 
$p^H_{res} \circ \phi = f^H$. 
Let $f^G_{ext}$ and $f^H_{ext}$ be some anti-twinned extensions of $f^G$ and $f^H$, respectively. That means, $AT(\overrightarrow{P})$
admits anti-twinned homomorphisms $f^G_{ext}$ and $f^H_{ext}$ 
to $AT(\overrightarrow{G})$ and $AT(\overrightarrow{H})$, respectively.
That means, there exists a unique homomorphism $\psi: AT(\overrightarrow{P}) \rightarrow AT(\overrightarrow{G}) \times 
AT(\overrightarrow{H})$ satisfying 
$p^G \circ \psi = f^G_{ext}$ and 
$p^H \circ \psi = f^H_{ext}$. 
Moreover, since 
$\psi'(v) = (f^G_{ext}(v), f^H_{ext}(v))$ satisfies the above
conditions, due to the uniqueness of $\psi$, we have $\psi = \psi'$. Notice that $\psi = \psi'$ is an anti-twinned homomorphism. That means, its anti-twinned restriction 
$\psi_{res}$ satisfies the conditions required to be satisfied by $\phi$. That is, we can now fix $\phi = \psi_{res}$. 
Clearly, $p^G_{res} \circ \phi = f^G$ and 
$p^H_{res} \circ \phi = f^H$.  

That means, $\overrightarrow{G} \times_p \overrightarrow{H}$ 
satisfies all but the uniqueness (up to pushable isomorphism) condition of a categorical product. However, since any other 
directed graph $\overrightarrow{Q}$ satisfying all but the uniqueness property of categorical product in the category of 
pushable directed graphs will also imply that 
$AT(\overrightarrow{Q})$ satisfies all properties (except uniqueness) the categorical product of $AT(\overrightarrow{G})$ and $AT(\overrightarrow{H})$. Due to the uniqueness of $AT(\overrightarrow{G}) \times AT(\overrightarrow{H})$, we have $AT(\overrightarrow{Q}) = AT(\overrightarrow{G}) \times AT(\overrightarrow{H}) = AT(\overrightarrow{G} \times_p \overrightarrow{H})$, which  implies  $\overrightarrow{Q}$ is pushably isomorphic to 
$\overrightarrow{G} \times_p \overrightarrow{H}$ using Theorem~\ref{th AT isomorphic}.  
\end{proof}

Since categorical product exists in the category of pushable directed graphs, it is natural to define and study its multiplicative objects. A directed graph $\overrightarrow{K}$ is \textit{pushably multiplicative} if $\overrightarrow{G} \times_p \overrightarrow{H} \xrightarrow{push} \overrightarrow{K}$
implies 
$\overrightarrow{G} \xrightarrow{push} \overrightarrow{K}$ or $\overrightarrow{H} \xrightarrow{push} \overrightarrow{K}$. 
A directed graph is \textit{pushably non-multiplicative} if it is not pushably multiplicative. 
Thus, the natural question in this context 
is as follows.

\begin{question}\label{ques pushably digraph multiplicative}
     Can you characterize all pushably multiplicative directed graphs? 
\end{question}

  This question was unexplored, and we study it for the first time. However, our results in this context provides new infinite families of non-multiplicative directed graphs. We summarize our main contributions in Section~\ref{sec pushable multiplicativity}.

\section{Pushable exponential directed graphs}\label{sec pushable exponential}
We study Question~\ref{ques pushably digraph multiplicative} for the families of bipartite directed graphs, oriented cycles and transitive tournaments; do a complete characterization based on if they are pushably multiplicative or not. For our proof, we prove the existence of exponential objects in the category of pushable directed graphs and use it as a tool for our proofs, and in the process solve the following open question posed in~\cite{das2026update}.

\begin{question}[Das, P.D., Sen and Taruni 2026~\cite{das2026update}]\label{ques push exponential}
    Does there exist exponential objects in the category of pushable directed graphs? 
\end{question}

Let us recall the exponential object in the category of directed graphs. Given two directed graphs $\overrightarrow{K}$ and $\overrightarrow{G}$, the \textit{exponential directed graph} $\overrightarrow{K}^{\overrightarrow{G}}$
is defined as follows: 
the vertices of $\overrightarrow{K}^{\overrightarrow{G}}$
are functions $\phi: V(\overrightarrow{G}) \to V(\overrightarrow{K})$ and there is an arc from  
such a vertex $\phi$ to another vertex $\psi$ if 
for any arc $uv$ of $\overrightarrow{G}$, 
 $\phi(u)\psi(v)$ is an arc of $\overrightarrow{K}$. 
 It is known~\cite{hell-nesetril-book} that exponential directed graph $\overrightarrow{K}^{\overrightarrow{G}}$ is the exponential object in the category of directed graphs. 

We show that the exponent of an anti-twinned directed graph
to another anti-twinned directed graph is itself 
an anti-twinned directed graph. This forms a major step of our work. 

\begin{lemma}\label{lem exp key}
    Given two directed graphs $\overrightarrow{K}$ and $\overrightarrow{G}$, the directed graph 
    $AT(\overrightarrow{K})^{AT(\overrightarrow{G})}$ is also an anti-twinned directed graph. 
\end{lemma}

\begin{proof}
Notice that, it is enough to find appropriate anti-twins for every vertex of 
$AT(\overrightarrow{K})^{AT(\overrightarrow{G})}$. 
Let $\phi$ be a vertex of $AT(\overrightarrow{K})^{AT(\overrightarrow{G})}$, that is, 
$\phi: V(AT(\overrightarrow{G})) \to V(AT(\overrightarrow{K}))$ is a function. Let $\phi^*$ be given by the following rule
$$\phi^*(v) = \phi(v^*)^*.$$ 
We claim that $\phi$ and $\phi^*$ are anti-twins of each other. For the sake of consistency, we must have $(\phi^*)^* = \phi$. Let us first check this. 
Observe that 
$$\phi^{**}(v) = (\phi^*(v^*))^* = ((\phi(v^{**}))^*)^* = \phi(v)^{**} = \phi(v).$$
Thus, our claim is consistent.

Next suppose $\phi \psi$ is an arc in $AT(\overrightarrow{K})^{AT(\overrightarrow{G})}$. 
Note that it is enough to prove that $\psi \phi^{*}$ is an arc 
in $AT(\overrightarrow{K})^{AT(\overrightarrow{G})}$. 
That means, we need to show that 
for any arc 
$uv$ in $AT(\overrightarrow{G})$, $\psi(u) \phi^*(v)$ is an arc of $AT(\overrightarrow{K})$. 
In $AT(\overrightarrow{G})$, $uv$ is an arc, which implies 
$v^*u$ is also an arc. That means, $\phi(v^*)\psi(u)$ is an arc in $AT(\overrightarrow{K})$. This implies,  $\psi(u)\phi(v^*)^*$ is an arc. Since $\phi^*(v)=\phi(v^*)^*$,  $\psi(u) \phi^*(v)$ is an arc in $AT(\overrightarrow{K})^{AT(\overrightarrow{G})}$. This completes the proof. 
\end{proof}

Let $[\overrightarrow{K}^{\overrightarrow{G}}]$ be a directed graph satisfying 
\begin{equation}\label{eq exponential antitwin}
    AT([\overrightarrow{K}^{\overrightarrow{G}}]) = AT(\overrightarrow{K})^{AT(\overrightarrow{G})}.
\end{equation}
We know that such a directed graph exists due to Lemma~\ref{lem exp key}, and it is unique (up to pushable isomorphism) due to Theorem~\ref{th AT isomorphic}. 
Thus, $[\overrightarrow{K}^{\overrightarrow{G}}]$ we will refer to the above mentioned unique directed graph that satisfies Eq$^n$~\ref{eq exponential antitwin}. 
This directed graph turns out to be the exponential object in the category of pushable directed graph, which we will show in the following. 

\begin{theorem}\label{thm push-exponential}
    Let $\overrightarrow{K}$ and $\overrightarrow{G}$ be directed graphs. 
    Let $[\overrightarrow{K}^{\overrightarrow{G}}]$ be the directed graph that satisfies Eq$^n$~\ref{eq exponential antitwin}. Then $[\overrightarrow{K}^{\overrightarrow{G}}]$ is the exponential object of $\overrightarrow{K}$ to the 
    $\overrightarrow{G}$ in the category of pushable directed graphs.    
\end{theorem}

\begin{proof}
    Suppose for an arbitrary directed graph $\overrightarrow{H}$ a pushable homomoprhism 
    $g: \overrightarrow{H} \times \overrightarrow{G} \xrightarrow{push} \overrightarrow{K}$ is given.
 We need to prove the existence of a pushable homomorphism $\lambda: \overrightarrow{H} \xrightarrow{push} [\overrightarrow{K}^{\overrightarrow{G}}]$, 
such that the following diagram commutes. 
\begin{figure}[H]
    \centering
    \tikzset{
    mid-isharrow/.style={
        postaction={
            decorate,
            decoration={
                markings,
                mark=at position 0.7 with {
                    \arrow{Latex[length=2mm]}
                }
            }
        }
    }
}

\tikzset{
    endarrow/.style={
        postaction={
            decorate,
            decoration={
                markings,
                mark=at position 1 with {
                    \arrow{Latex[length=2mm]}
                }
            }
        }
    }
}

\tikzset{
  every label/.style={fill=white,inner sep=1pt}
}

\usetikzlibrary{backgrounds}

\pgfdeclarelayer{bg}
\pgfdeclarelayer{fg}
\pgfsetlayers{bg,main,fg}

\begin{tikzpicture}[scale=1,
vertex/.style={circle,fill=black,inner sep=1.5pt},
every label/.style={font=\small, fill=white!50, inner sep=1pt},
label distance=4pt]

    \begin{scope}[xshift=0cm]
        
        \node (h) at (0,2) {$\overrightarrow{H}$};
        \node (x) at (0,0) {$\left[ \overrightarrow{K}^{\overrightarrow{G}} \right]$};

        \draw[endarrow] (h)--node[left]{$\lambda$}node[right]{$push$}(x);     
    
    \end{scope}

    \begin{scope}[xshift=3cm]

        \node (gh) at (0,2) {$\overrightarrow{H}\times\overrightarrow{G}$};
        \node (kg) at (0,0) {$\left[ \overrightarrow{K}^{\overrightarrow{G}} \right]\times \overrightarrow{G}$};
        \node (k) at (4,0) {$\overrightarrow{K}$};

        \draw[endarrow,dashed] (gh)--node[left]{$\lambda\times 1_{\overrightarrow{G}}$}node[right]{$push$}(kg);
        \draw[endarrow] (gh)--node[above,sloped]{$push$}node[below,sloped]{$g$}(k);
        \draw[endarrow] (kg)--node[left,below]{$\epsilon$}node[left,above]{$push$}(k);
        
    \end{scope}

\end{tikzpicture}
\end{figure}
\noindent Here the evaluation function is defined as $\epsilon((\phi,v)) = \pi(\phi(v))$ where $\pi: AT(\overrightarrow{K}) \xrightarrow{push} \overrightarrow{K}$ is given by
$$
\pi(v) = \begin{cases}
			v, & \text{if $v \in V(\overrightarrow{K})$,}\\
            v^*, & \text{if $v \in V^*(\overrightarrow{K})$}.
		 \end{cases}
$$

It is known that the exponential object of $AT(\overrightarrow{K})$ to the $AT(\overrightarrow{G})$ is 
$AT(\overrightarrow{K})^{AT(\overrightarrow{G})}$ in the category of directed graphs. Let $g_{ext}$ be an anti-twinned extension of $g$. Then the universal properties of the exponential object will imply the existence of a homomorphism $\lambda_{ext}$ (say) 
of $AT(\overrightarrow{H})$ to 
$AT(\overrightarrow{K})^{AT(\overrightarrow{G})} = AT([\overrightarrow{K}^{\overrightarrow{G}}])$
such that the following diagraph commutes. 
\begin{figure}[H]
    \centering
    \tikzset{
    mid-isharrow/.style={
        postaction={
            decorate,
            decoration={
                markings,
                mark=at position 0.7 with {
                    \arrow{Latex[length=2mm]}
                }
            }
        }
    }
}

\tikzset{
    endarrow/.style={
        postaction={
            decorate,
            decoration={
                markings,
                mark=at position 1 with {
                    \arrow{Latex[length=2mm]}
                }
            }
        }
    }
}

\tikzset{
  every label/.style={fill=white,inner sep=1pt}
}

\usetikzlibrary{backgrounds}

\pgfdeclarelayer{bg}
\pgfdeclarelayer{fg}
\pgfsetlayers{bg,main,fg}

\begin{tikzpicture}[scale=1,
vertex/.style={circle,fill=black,inner sep=1.5pt},
every label/.style={font=\small, fill=white!50, inner sep=1pt},
label distance=4pt]

    \begin{scope}[xshift=0cm]
        
        \node (h) at (0,2) {$AT(\overrightarrow{H})$};
        \node (x) at (0,0) {$AT(\left[ \overrightarrow{K}^{\overrightarrow{G}} \right])$};

        \draw[endarrow] (h)--node[left]{$\lambda_{ext}$}(x);     
    
    \end{scope}

    \begin{scope}[xshift=5cm]

        \node (gh) at (0,2) {$AT(\overrightarrow{H})\times AT(\overrightarrow{G})$};
        \node (kg) at (0,0) {$AT\left(\left[ \overrightarrow{K}^{\overrightarrow{G}} \right]\right)\times AT(\overrightarrow{G})$};
        \node (k) at (4,0) {$AT(\overrightarrow{K})$};

        \draw[endarrow,dashed] (gh)--node[left]{$\lambda_{ext}\times 1_{AT(\overrightarrow{G})}$}(kg);
        \draw[endarrow] (gh)--node[above,yshift=2pt]{$g_{ext}$}(k);
        \draw[endarrow] (kg)--node[left,below]{$\epsilon_{ext}$}(k);
        
    \end{scope}

\end{tikzpicture}
\end{figure}
\noindent Here the evaluation function is defined in a natural way, that is, $\epsilon_{ext}((\phi,v)) = \phi(v)$ for all $(\phi, v) \in AT([\overrightarrow{K}^{\overrightarrow{G}}]) \times AT(\overrightarrow{G})$.

Notice that $\epsilon_{ext}$  is an anti-twinned homomorphism, and $\epsilon$ is its anti-twinned 
restriction. Moreover, the identity homomorphism 
$1_{AT(\overrightarrow{G})}: AT(\overrightarrow{G}) \rightarrow AT(\overrightarrow{G})$ is given by 
$1_{AT(\overrightarrow{G})}(v) = v$ for all $v \in V(AT(\overrightarrow{G}))$. Similarly, the identity pushable homomorphism  
$1_{\overrightarrow{G}}: \overrightarrow{G} \xrightarrow{push} \overrightarrow{G})$ is given by 
$1_{\overrightarrow{G}}(v) = v$ for all $v \in V(\overrightarrow{G})$.
Hence, notice that,
 $1_{AT(\overrightarrow{G})}$  is an anti-twinned homomorphism, and $1_{\overrightarrow{G}}$ is its anti-twinned restriction.

Next we want to show that $\lambda_{ext}$ is an 
anti-twinned homomorphism. Let $x$ be a vertex of 
$AT(\overrightarrow{H})$. To verify that 
$\lambda_{ext}$ is an anti-twinned homomorphism, we need to check that 
$\lambda_{ext}(x^*) = \lambda_{ext}(x)^*$. 
For convenience, suppose that 
$\lambda_{ext}(x) = \phi$ and 
$\lambda_{ext}(x^*) = \psi$. 
Essentially, we need to show that 
$\phi(v^*)^* = \psi(v)$ 
as per Lemma~\ref{lem exp key}. 
Notice that 
$$\phi(v^*) = (\lambda_{ext}(x))(v^*) = \epsilon \circ (\lambda_{ext} \times 1_{AT(\overrightarrow{G})})(x,v^*) = g_{ext}(x,v^*).$$
Furthermore,
$$\psi(v) = (\lambda_{ext}(x^*))(v) = \epsilon \circ (\lambda_{ext} \times 1_{AT(\overrightarrow{G})})(x^*,v) = g_{ext}(x^*,v).$$
Since $g_{ext}$ is an anti-twinned homomorphism, 
and since $(x,v^*)$ and $(x^*, v)$ are anti-twins of each other, we have 
$$g_{ext}(x,v^*)^* = g_{ext}(x^*, v).$$
This implies, $\phi(v^*)^* = \psi(v)$, and thus, 
$\lambda_{ext}$ is an anti-twinned homomorphism. 

By taking $\lambda$ as the anti-twinned restriction of $\lambda_{ext}$, we can show that the first diagram commutes. 
Moreover, if any other directed graph, say, $\overrightarrow{Q}$, satisfies 
the conditions of an exponential object of $\overrightarrow{K}$ to the $\overrightarrow{G}$ 
in the category of pushable directed graphs other than $[\overrightarrow{K}^{\overrightarrow{G}}]$, 
then $AT(\overrightarrow{Q})$ will satisfy the conditions of an exponential object of $AT(\overrightarrow{K})$ to the $AT(\overrightarrow{G})$ 
in the category of directed graphs. However, since we know that $AT(\overrightarrow{K})^{AT(\overrightarrow{G})}$ is unique, we must have $AT(\overrightarrow{Q}) = AT(\overrightarrow{K})^{AT(\overrightarrow{G})} =
AT([\overrightarrow{K}^{\overrightarrow{G}}])$. Thus, $\overrightarrow{Q}$ is pushably isomorphic to 
$[\overrightarrow{K}^{\overrightarrow{G}}]$ due to Theorem~\ref{th AT isomorphic}. 
This completes the proof. 
\end{proof}

Since we have proved the uniqueness of $[\overrightarrow{K}^{\overrightarrow{G}}]$, let us call it the  \textit{pushable exponential directed graph} from now on. 

Let us establish one of the most important tools now. 

\begin{lemma}\label{lem push-multiplicative}
    A directed graph $\overrightarrow{K}$ is pushably multiplicative if and only if one of the following conditions is satisfied:
    \begin{enumerate}[(i)]
        \item $\overrightarrow{G} \centernot{\xrightarrow{push}} \overrightarrow{K}$ implies 
        $\left[\overrightarrow{K}^{\overrightarrow{G}}\right] \xrightarrow{push} \overrightarrow{K}$,

        \item $\overrightarrow{G} \centernot{\xrightarrow{push}} \overrightarrow{K}$ implies 
        $AT(\overrightarrow{K})^{AT(\overrightarrow{G})} \xrightarrow{push} \overrightarrow{K}$.
    \end{enumerate}
\end{lemma}

\begin{proof}
\noindent $(i)$ Assume that $\overrightarrow{K}$ is pushably multiplicative.
   Thus, if $\overrightarrow{G} \centernot{\xrightarrow{push}} \overrightarrow{K}$ and $\left[\overrightarrow{K}^{\overrightarrow{G}}\right] \centernot{\xrightarrow{push}} \overrightarrow{K}$, then $\left[\overrightarrow{K}^{\overrightarrow{G}}\right] \times \overrightarrow{G}\centernot{\xrightarrow{push}} \overrightarrow{K}$.
    However, we have $\epsilon: \left[\overrightarrow{K}^{\overrightarrow{G}}\right] \times \overrightarrow{G}\xrightarrow{push} \overrightarrow{K}$ by Theorem~\ref{thm push-exponential}, which is a contradiction.

    Conversely, 
    suppose that condition $(i)$ is satisfied. 
    Also suppose $\overrightarrow{G} \times_p \overrightarrow{H} \xrightarrow{push} \overrightarrow{K}$. 
    It is enough to show that 
    either $\overrightarrow{G} \xrightarrow{push} \overrightarrow{K}$ or 
    $\overrightarrow{H} \xrightarrow{push} \overrightarrow{K}$. 
    Therefore, if $\overrightarrow{G} \xrightarrow{push} \overrightarrow{K}$, then we are done. If 
    $\overrightarrow{G} \centernot{\xrightarrow{push}} \overrightarrow{K}$,
    then $\left[\overrightarrow{K}^{\overrightarrow{G}}\right] \xrightarrow{push} \overrightarrow{K}$ due to condition $(i)$. Moreover, due to the universal properties of an exponential object, 
    $\overrightarrow{H} \xrightarrow{push} \left[\overrightarrow{K}^{\overrightarrow{G}}\right]$. Composing the two above pushable homomorphisms, we have $\overrightarrow{H} \xrightarrow{push} \overrightarrow{K}$. Therefore, $\overrightarrow{K}$ is pushably multiplicative.

     \medskip

     \noindent $(ii)$ This follows from the fact 
$AT(\left[\overrightarrow{K}^{\overrightarrow{G}}\right]) = AT(\overrightarrow{K})^{AT(\overrightarrow{G})}$ due to Lemma~\ref{lem exp key} and the definition of $\left[\overrightarrow{K}^{\overrightarrow{G}}\right]$. 
\end{proof}

In particular, we will use the following corollary of the above lemma as a tool in our proofs.

\begin{corollary}
    If there exists a directed graph $\overrightarrow{G}$
    such that $\overrightarrow{G} \centernot{\xrightarrow{push}} \overrightarrow{K}$ and 
    $AT(\overrightarrow{K})^{AT(\overrightarrow{G})} \centernot{\xrightarrow{push}} \overrightarrow{K}$, then $\overrightarrow{K}$ is not pushably multiplicative. 
\end{corollary}

\section{Pushable multiplicativity of directed graphs}\label{sec pushable multiplicativity}
A directed graph is \textit{balanced} if it is possible to push some of its vertices and convert every vertex into a source or a sink. If a directed graph is not balanced, then it is \textit{unbalanced}.

\begin{theorem}\label{thm push-oriented-bipartite}
    A bipartite oriented graph is pushably multiplicative if and only if it is balanced. 
\end{theorem}

\begin{proof}
   Let $\overrightarrow{K}_2$ denote an orientation of $K_2$.     
    It is known~\cite{das2026update} that an oriented graph is pushably homomorphically equivalent to $\overrightarrow{K}_2$ if and only if it is balanced. 
    That means, to prove that all balanced bipartite directed graph are pushably multiplicative, it is enough to prove that 
    $\overrightarrow{K}_2$ is multiplicative.

    We know that any odd oriented cycle is pushably equivalent to 
    the directed cycle~\cite{das2026update}.
    We also know~\cite{das2026update} that
    an even unbalanced  cycle on $4k$ vertices, for some $k \geq 1$, is equivalent to the oriented cycle whose all but one arc is a forward arc and an  even unbalanced  cycle on $4k+2$ vertices, for some $k \geq 1$, is equivalent to the directed cycle.

    Let $\overrightarrow{UC}_{4k}$ be an unbalanced even cycle on $4k$ vertices for some $k \geq 1$. Suppose the vertices of 
$\overrightarrow{UC}_{4k}$ are $\{v_1, v_2, \ldots, v_{4k}\}$. Also, assume that its arcs are $v_1v_2, v_2v_3, \ldots, v_{4k-1}v_{4k}$, and $v_1v_{4k}$. Note that, in the anti-twinned graph  $AT(\overrightarrow{UC}_{4k})$
the following directed cycle of length $4k+2$
is present:
$$v_1v_2 \ldots v_{4k} v^*_{4k-1} v^*_{4k} v_1.$$

That means, anti-twin of any unbalanced directed graph contains a directed cycles of length 
$2k+1$ or $4k+2$, for some $k \geq 1$. Furthermore, note that, the categorical product of two directed cycles of lengths $k_1$ and $k_2$ contains a directed cycle of length equal to the L.C.M. of $k_1$ and $k_2$. Observe that the L.C.M. of two numbers of the form $2k+1$ or $4k+2$ is also of the form $2k+1$ or $4k+2$. 
Suppose $\overrightarrow{G}$ and $\overrightarrow{H}$ are two unbalanced directed graphs. Thus, we can say that  the categorical product  $AT(\overrightarrow{G}) \times AT(\overrightarrow{H})$ is an unbalanced directed graph. As any directed graph is pushably homomorphically  equivalent its anti-twinned directed graph, we can conclude that the pushable categorical product of two unbalanced graph is unbalanced, and thus does not admit a pushable homomorphism to $\overrightarrow{K}_2$. This implies that $\overrightarrow{K}_2$ is pushably multiplicative.

\medskip

Next we will show that all unbalanced bipartite directed graphs are pushably non-multiplicative.
Let $\overrightarrow{K}$ be an unbalanced bipartite directed graph, and let $2k+2$ be the length of its shortest unbalanced even oriented cycle. Notice that, $k \geq 1$ as we are working on oriented graphs.

Let $\overrightarrow{D}_2$ denote the graph on $2$ vertices $x, y$ having two arcs $xy$ and $yx$. Let $\overrightarrow{DC}_{2k+1}$ denote the directed cycle on $2k+1$ vertices. 
Notice that $\overrightarrow{D}_2 \centernot{\xrightarrow{push}} \overrightarrow{K}$ since $\overrightarrow{K}$ is an oriented graph, and  
$\overrightarrow{DC}_{2k+1} \centernot{\xrightarrow{push}} \overrightarrow{K}$ since 
$\overrightarrow{K}$ is bipartite. 

Note that  $\overrightarrow{D}_2 \times_p \overrightarrow{DC}_{2k+1} = \overrightarrow{D}_2 \times \overrightarrow{DC}_{2k+1}$, since all vertices of $\overrightarrow{D}_2$ are push invariant. Furthermore, $\overrightarrow{D}_2 \times \overrightarrow{DC}_{2k+1}$ is nothing but 
the directed cycle $\overrightarrow{DC}_{4k+2}$ on $4k+2$ vertices, that is the unbalanced even cycle on $4k+2$ vertices. 
Observe that $\overrightarrow{DC}_{4k+2}$ admits a homomorphism to any unbalanced even cycle of a lesser length. In particular, since $\overrightarrow{K}$ contains an unbalanced cycle of lesser length, $\overrightarrow{DC}_{4k+2} \xrightarrow{push} \overrightarrow{K}$. This implies that $\overrightarrow{K}$ is pushably non-multiplicative. 
\end{proof}

\begin{theorem}\label{thm push-oriented-cycle}
    An oriented cycle is pushably multiplicative if and only if it is a balanced and even.
\end{theorem}

\begin{proof}
    We have already shown that even oriented cycles are pushably multiplicative if and only if they are balanced in Theorem~\ref{thm push-oriented-bipartite}. Thus, it is enough to prove that directed odd cycles are pushably non-multiplicative, since all oriented odd cycles are equivalent to the directed cycle. In other words, it is enough to prove that the directed cycle 
    $\overrightarrow{DC}_{2k+1}$ is pushably non-multiplicative for any $k \geq 1$. 
    For convenience assume that $\overrightarrow{DC}_{2k+1}$ is the following directed cycle:
    $0,1,\cdots, \cdots, 2k, 0$.

Let $\overrightarrow{D}_2$ denote the directed 
graph on $2$ vertices $x$ and $y$ having two 
arcs $xy$ and $yx$. 
Since both $x$ and $y$ are push invariant, 
$AT(\overrightarrow{D}_2) = \overrightarrow{D}_2$. Clearly, $\overrightarrow{D}_2$ does not admit a pushable homomorphism to $\overrightarrow{DC}_{2k+1}$. Thus, if we can show that 
the graph $AT(\overrightarrow{DC}_{2k+1})^{\overrightarrow{D}_2}$ does not admit a homomorphism to  $\overrightarrow{DC}_{2k+1}$, then using Lemma~\ref{lem push-multiplicative} we can conclude that 
$\overrightarrow{DC}_{2k+1}$ is pushably non-multiplicative.

Indeed, we are going to show that a particular subgraph of $AT(\overrightarrow{DC}_{2k+1})^{\overrightarrow{D}_2}$ does not admit a homomorphism to 
$\overrightarrow{DC}_{2k+1}$. As a first step, let us discuss a convention to describe the above-mentioned subgraph. Notice that, the vertices of $AT(\overrightarrow{DC}_{2k+1})^{\overrightarrow{D}_2}$ are functions from $V(\overrightarrow{D}_2)$ to $V(AT(\overrightarrow{DC}_{2k+1}))$. 
If $f$ is such a function, then as a convention let us denote it by $(f(x),f(y))$. Notice that it is enough to know the values of $f(x)$ and $f(x)$ to know the whole function. 

Using this convention, in Fig.~\ref{fig AT-directed-cycle} we have presented a subgraph of $AT(\overrightarrow{DC}_{2k+1})^{\overrightarrow{D}_2}$. Let us call this graph $\overrightarrow{Z}$. 
It is possible to verify that this graph is indeed a subgraph of  $AT(\overrightarrow{DC}_{2k+1})^{\overrightarrow{D}_2}$.  While verifying, one may use the following observation for convenience:  there is an arc from 
$(a,b)$ to $(c,d)$ in $\overrightarrow{Z}$
if $ad, bc$ are arcs in $AT(\overrightarrow{DC}_{2k+1})$.

Now all that remains is to prove that 
$\overrightarrow{Z}$ does not admit a pushable homomorphism to $\overrightarrow{DC}_{2k+1}$. 
Let us first try to understand the structure of $\overrightarrow{Z}$, beginning with a directed path decomposition. We will list down some directed paths of $\overrightarrow{Z}$ which partitions its arcs.  
\begin{itemize}
    \item $\overrightarrow{P}_{L1} = (2k,0), (1,(2k-1)^*), ((2k-2),2), \dots ,((2k-1),1^*), (0,2k)$,

    \item $\overrightarrow{P}_{L2} = (0,2k), (0,1), (2,1) \dots (2k,(2k-1)), (2k,0)$,

    \item $\overrightarrow{P}_{R1} = (0,2k^*), ((2k-1),1), (2,(2k-2)^*), \dots ,(1,(2k-1)), (2k,0^*)$,

    \item $\overrightarrow{P}_{R2} = (0,2k^*), ((2k-1),2k^*), ((2k-1),(2k-1)^*), \dots ,(1,0^*), (2k,0^*)$,

    \item $\overrightarrow{P}_U = (0,2k),(0,2k^*)$,

    \item $\overrightarrow{P}_D = (2k,0^*),(2k,0)$.
\end{itemize}

Observe that the above-listed directed paths have the following size (number of arcs).

\begin{table}[h]
\centering
\renewcommand{\arraystretch}{1.7}
\begin{tabular}{|c|c|}
\hline
\textbf{Name of directed path} & \textbf{Size} \\
\hline
$\overrightarrow{P}_{L1}$ & $2k$ \\
\hline
$\overrightarrow{P}_{L2}$ & $2k + 1$ \\
\hline
$\overrightarrow{P}_{R1}$ & $2k$ \\
\hline
$\overrightarrow{P}_{R2}$ & $2k + 1$ \\
\hline
$\overrightarrow{P}_{U}$  & $1$ \\
\hline
$\overrightarrow{P}_{D}$  & $1$ \\
\hline
\end{tabular}
\end{table}

Next note that there are the following oriented cycles in $\overrightarrow{Z}$, all induced by a collection of the above mentioned directed paths.  
\begin{itemize}
    \item  $\overrightarrow{C}_{M} = \overrightarrow{P}_U, \overrightarrow{P}_{R1}, \overrightarrow{P}_{D}, \overrightarrow{P}_{L1}$, 

    \item $\overrightarrow{C}_{L} = \overrightarrow{P}_{L1},\overrightarrow{P}_{L2}$,

    \item $\overrightarrow{C}_{R} = \overrightarrow{P}_{R1},\overrightarrow{P}_{R2}$.
\end{itemize}
Observe that the above-listed oriented cycles have the following orders, number of forward arcs, number of backward arcs. Here since all the three oriented cycles are already embedded in a planar way in Fig.~\ref{fig AT-directed-cycle}, the arcs oriented clockwise will be counted as forward, and the arcs oriented counter-clockwise will be counted as backward. 

\begin{table}[h]
\centering
\renewcommand{\arraystretch}{1.7}
\begin{tabular}{|c|c|c|c|}
\hline
\textbf{Name of oriented cycle} & \textbf{Order} & \textbf{Forward arcs} & \textbf{Backward arcs} \\
\hline
$\overrightarrow{C}_{M}$ & $4k + 2$ & $4k + 2$ & $0$\\
\hline
$\overrightarrow{C}_{L}$ & $4k + 1$ & $0$ & $4k + 1$\\
\hline
$\overrightarrow{C}_{R}$ & $4k + 1$ & $2k + 1$ & $2k$\\
\hline
\end{tabular}
\end{table}

Suppose, $\overrightarrow{C}$ be an oriented cycle 
having $n_1$ forward arcs and $n_2$ backward arcs. The oriented cycle admits a 
homomorphism to $DC_{2k+1}$ if and only if $(2k+1)$ 
divides $(n_1-n_2)$. Notice that, if you allow to 
push the vertices of $\overrightarrow{C}$, then the 
number of forward and backward arcs will change yet 
their parity will be the same. We will use these two 
observations in the following.

Let us assume that $\overrightarrow{Z}'$ is a push 
equivalent directed graph of $\overrightarrow{Z}$ 
that admits a homomorphism to 
$\overrightarrow{DC}_{2k+1}$. Let us assume that the oriented cycles $\overrightarrow{C}_{M}, \overrightarrow{C}_{L}, \overrightarrow{C}_{R}$ got converted into $\overrightarrow{C}'_{M}, \overrightarrow{C}'_{L}, \overrightarrow{C}'_{R}$, respectively, in $\overrightarrow{Z}'$.

Suppose that  $\overrightarrow{C}'_{M}$ has
$n_1$ forward and $n_2$ backward arcs. 
Since $\overrightarrow{C}_{M}$ had $4k+2$ forward and $0$ backward arcs, $n_1, n_2$ must be even. Moreover, 
 we must have 
$n_1-n_2 \equiv 0~(mod~2k+1)$. That means the only two feasible solutions are
$(n_1, n_2) = (4k+2, 0)$ and 
$(0, 4k+2)$.  That means, $\overrightarrow{C}'_{M}$ is a directed cycle. Let us assume that all its arcs are forward (the other case, when all its arcs are backward, can be handled in a similar way, hence is omitted from the proof).

Notice that $\overrightarrow{C}_{L}$ has $0$ forward and $4k+1$ backward arcs. That means, $\overrightarrow{C}'_{L}$ has even number of forward and odd number of backward arcs. 
Moreover,
$\overrightarrow{C}'_{L}$ has at least $2k$ backward arcs (as they are part of $\overrightarrow{C}'_{M}$). 
Suppose $\overrightarrow{C}'_{L}$ have $n_1$ forward and $n_2$ backward arcs, then we must have 
$n_1-n_2 \equiv 0~(mod~2k+1)$. 
Observe that there is only one solutions that will satisfy all the constraints, namely, 
$(n_1, n_2) = (k, 3k+1)$ where $k$ must be even. 
 
Notice that $\overrightarrow{C}_{R}$ has $2k+1$ forward and $2k$ backward arcs. That means, $\overrightarrow{C}'_{R}$ has odd number of forward and even number of backward arcs. 
Moreover,
$\overrightarrow{C}'_{R}$ has at least $2k$ backward arcs (as they are part of $\overrightarrow{C}'_{M}$). 
Suppose $\overrightarrow{C}'_{L}$ have $n_1$ forward and $n_2$ backward arcs, then we must have 
$n_1-n_2 \equiv 0~(mod~2k+1)$. 
Observe that there is only one solutions that will satisfy all the constraints, namely, 
$(n_1, n_2) = (k, 3k+1)$ where $k$ must be odd.

This is a contradiction as $k$ cannot be simultaneously even and odd. Thus, $\overrightarrow{Z}$ does not admit a pushable homomorphism to $\overrightarrow{DC}_{2k+1}$. 
\end{proof}

\begin{figure}
    \centering
    \tikzset{
    mid-isharrow/.style={
        postaction={
            decorate,
            decoration={
                markings,
                mark=at position 0.7 with {
                    \arrow{Latex[length=2mm]}
                }
            }
        }
    }
}

\tikzset{
    endarrow/.style={
        postaction={
            decorate,
            decoration={
                markings,
                mark=at position 1 with {
                    \arrow{Latex[length=2mm]}
                }
            }
        }
    }
}

\tikzset{
  every label/.style={fill=white,inner sep=1pt}
}

\usetikzlibrary{backgrounds}

\pgfdeclarelayer{bg}
\pgfdeclarelayer{fg}
\pgfsetlayers{bg,main,fg}

\begin{tikzpicture}[scale=0.7,
vertex/.style={circle,fill=black,inner sep=1.5pt},
every label/.style={font=\small, fill=white!50, inner sep=1pt},
label distance=4pt]

    \begin{scope}[xshift=0cm]
        
        \node (a) at (-1,0) {$(0, 1)$};
        \node (b) at (3,0) {$(0, 2k)$};
        \node (c) at (7,0) {$(0, 2k^*)$};
        \node (d) at (11,0) {$((2k-1), 2k^*)$};
        
        \draw[endarrow] (b) -- (a);
        \draw[endarrow] (b) -- (c);
        \draw[endarrow] (c) -- (d);

        \node (l1) at (-1,-2) {$(2, 1)$};
        \node (l2) at (-1,-4) {};
        \node (l4) at (-1,-6) {$((2k-2),(2k-1))$};
        \node (l5) at (-1,-8) {$(2k, (2k-1))$};
        
        \draw[endarrow] (a) -- (l1);
        \draw[endarrow,dotted] (l1) -- (l4);
        \draw[endarrow] (l4) -- (l5);

        \node (m1) at (7,-2) {$((2k-1), 1)$};
        \node (m2) at (7,-4) {$(2, (2k-2)^*)$};
        \node (m3) at (7,-6) {$(1, (2k-1))$};
        \node (m4) at (7,-8) {$(2k, 0^*)$};

        \draw[endarrow] (c) -- (m1);
        \draw[endarrow] (m1) -- (m2);
        \draw[dotted,endarrow] (m2) -- (m3);
        \draw[endarrow] (m3) -- (m4);

        \node (r1) at (11,-2) {$((2k-1), (2k-2)^*)$};
        \node (r2) at (11,-6) {$(1, 2^*)$};
        \node (r3) at (11,-8) {$(1, 0^*)$};
        
        \draw[endarrow] (d) -- (r1);
        \draw[dotted,endarrow] (r1) -- (r2);
        \draw[endarrow] (r2) -- (r3);

        \node (b1) at (3,-8) {$(2k, 0)$};
        \node (b2) at (3,-6) {$(1, (2k-1)^*)$};
        \node (b3) at (3,-4) {$((2k-2), 2)$};
        \node (b4) at (3,-2) {$((2k-1), 1^*)$};

        \draw[endarrow] (b1) -- (b2);
        \draw[endarrow] (b2) -- (b3);
        \draw[endarrow,dotted] (b3) -- (b4);
        \draw[endarrow] (b4) -- (b);
        
        \draw[endarrow] (l5) -- (b1);
        \draw[endarrow] (m4) -- (b1);
        \draw[endarrow] (r3) -- (m4); 
    
    \end{scope}

\end{tikzpicture}
    \caption{A subgraph $\overrightarrow{Z}$ of $AT(\overrightarrow{DC}_{2k+1})^{\overrightarrow{D}_2}$.}
    \label{fig AT-directed-cycle}
\end{figure}

\begin{theorem}\label{thm push-oriented-TT}
    A transitive tournament $\overrightarrow{TT}_n$ on $n$ vertices is pushably multiplicative if and only if $n \leq 2$. 
\end{theorem}

\begin{proof}
    The transitive tournaments on $1$ and $2$ vertices are pushably multiplicative using 
    Theorem~\ref{thm push-oriented-bipartite} since they are balanced bipartite graphs. The transitive tournament 
    on $3$ vertices is pushably non-multiplicative due to Theorem~\ref{thm push-oriented-cycle} as it is pushably equivalent to the directed cycle on $3$ vertices.

    Thus, we are left with showing that transitive tournaments $\overrightarrow{TT}_n$ on $n$ vertices are pushably non-multiplicative for all $n \geq 4$. 

    Let $\overrightarrow{D}_2$ denote the directed 
graph on $2$ vertices $x$ and $y$ having two 
arcs $xy$ and $yx$. 
Since both $x$ and $y$ are push invariant, 
$AT(\overrightarrow{D}_2) = \overrightarrow{D}_2$. Clearly, $\overrightarrow{D}_2$ does not admit a pushable homomorphism to $\overrightarrow{TT}_{n}$ for any $n \geq 4$. 
Thus, if we can show that 
the graph $AT(\overrightarrow{TT}_{n})^{\overrightarrow{D}_2}$ does not admit a homomorphism to  $\overrightarrow{DC}_{2k+1}$, then using Lemma~\ref{lem push-multiplicative} we can conclude that 
$\overrightarrow{TT}_{n}$ is pushably non-multiplicative.

Indeed, we are going to show that a particular subgraph of $AT(\overrightarrow{TT}_{n})^{\overrightarrow{D}_2}$ does not admit a homomorphism to 
$\overrightarrow{TT}_{n}$, for any $n \geq 4$. As a first step, let us discuss a convention to describe the above-mentioned subgraph. Notice that, the vertices of $AT(\overrightarrow{TT}_{n})^{\overrightarrow{D}_2}$ are functions from $V(\overrightarrow{D}_2)$ to $V(AT(\overrightarrow{TT}_{n}))$. 
If $f$ is such a function, then as a convention let us denote it by $(f(x),f(y))$. Notice that it is enough to know the values of $f(x)$ and $f(x)$ to know the whole function. 
Moreover, as a convention, we can suppose that the vertices of $\overrightarrow{TT}_n$ are $0, 1, \ldots, n-1$, while its arcs are oriented from $i$ to $j$ where $i < j$. From this naming convention it is clear that if $\overrightarrow{F}$ is a subgraph of 
$AT(\overrightarrow{TT}_{4})^{\overrightarrow{D}_2}$, then it is also a subgraph of 
$AT(\overrightarrow{TT}_{n})^{\overrightarrow{D}_2}$ for all $n \geq 4$. We present such a subgraph $\overrightarrow{F}$ on $8$ vertices and show that it does not admit a pushable homomorphism to any $\overrightarrow{TT}_n$ for $n \geq 4$. That will conclude the proof.

First of all, let us present the special subgraph $\overrightarrow{F}$ in Fig.~\ref{fig F}. One can verify that $\overrightarrow{F}$ is indeed a subgraph of $AT(\overrightarrow{TT}_{4})^{\overrightarrow{D}_2}$, and thus a subgraph of 
$AT(\overrightarrow{TT}_{n})^{\overrightarrow{D}_2}$ for all $n \geq 4$. Thus, we are left to prove that $\overrightarrow{F}$ does not admit a pushable homomorphism to $\overrightarrow{TT}_n$ 
for any $n \geq 4$. 
Notice that, $\overrightarrow{F}$ admits a pushable homomorphism to $\overrightarrow{TT}_8$ 
if and only if it is possible to push 
some vertices of $\overrightarrow{F}$ and convert it into a directed acyclic graph. 
Let $\overrightarrow{F}'$ be a  directed acyclic graph push equivalent
to $\overrightarrow{F}$. 
Let $F_1 = \{(1^*,0), (3,2^*), (3^*,0^*), (1^*,2)\}$ and $F_2 = \{(1,0), (3^*,2^*), (3,0^*), (1^*,2^*)\}$. 
Observe that, the two induced tournaments $\overrightarrow{F}'[F_1]$ and $\overrightarrow{F}'[F_2]$ , respectively, must be transitive tournaments. 

We do an exhaustive case analysis to prove the impossibility of converting $\overrightarrow{F}$
into a directed acyclic graph $\overrightarrow{F}'$ via pushing some vertices. Suppose, $S$ is the set of vertices that we pushed to obtain $\overrightarrow{F}'$ from $\overrightarrow{F}$. Since pushing $S$ or $V(\overrightarrow{F}) \setminus S$ has the same effect, we can assume that $S$ contains the vertex $(3^*,0^*)$ in particular. Moreover, let $X = S \cap F_1$ and $Y = S \cap F_2$. The case analysis is based on the following assumptions and observations. 

\begin{figure}
    \centering
                \tikzset{
    mid-isharrow/.style={
        postaction={
            decorate,
            decoration={
                markings,
                mark=at position 0.7 with {
                    \arrow{Latex[length=2mm]}
                }
            }
        }
    }
}

\tikzset{
    endarrow/.style={
        postaction={
            decorate,
            decoration={
                markings,
                mark=at position 1 with {
                    \arrow{Latex[length=2mm]}
                }
            }
        }
    }
}

\tikzset{
  every label/.style={fill=white,inner sep=1pt}
}

\usetikzlibrary{backgrounds}

\pgfdeclarelayer{bg}
\pgfdeclarelayer{fg}
\pgfsetlayers{bg,main,fg}

\begin{tikzpicture}[scale=1,
vertex/.style={circle,fill=black,inner sep=1.5pt},
every label/.style={font=\small, fill=white!50, inner sep=1pt},
label distance=4pt]

        




    \begin{scope}[xshift=0cm]

        \node (10) at (0,0) {$(1,0)$};
        \node (3*2*) at (2,0) {$(3^*,2^*)$};
        \node (30*) at (4,0) {$(3,0^*$)};
        \node (1*2*) at (6,0) {$(1^*,2^*)$};
        
        \node (1*0) at (0,1) {$(1^*,0)$};
        \node (32*) at (2,1) {$(3,2^*)$};
        \node (3*0*) at (4,1) {$(3^*,0^*)$};
        \node (1*2) at (6,1) {$(1^*,2)$};

        \draw[endarrow] (1*0) to (32*);
        \draw[endarrow] (1*0) to (10);
        \draw[endarrow] (32*) to (3*2*);
        \draw[endarrow] (32*) to (3*0*);
        \draw[endarrow] (3*0*) to (1*2);
        \draw[endarrow,bend right=20] (3*0*) to (1*0);
        \draw[endarrow] (1*2) to (1*2*);
        \draw[endarrow,bend right=20] (1*2) to (32*);
        \draw[endarrow,bend right=20] (1*2) to (1*0);
        
        \draw[endarrow,bend left=20] (1*2*) to (3*2*);
        \draw[endarrow,bend left=20] (1*2*) to (10);
        \draw[endarrow] (30*) to (1*2*);
        \draw[endarrow] (30*) to (3*2*);
        \draw[endarrow] (30*) to (3*0*);
        \draw[endarrow] (3*2*) to (10);
        \draw[endarrow,bend right=20] (10) to (30*);

        \node at (7,1) {$F_1$};
        \node at (7,0) {$F_2$};

        \node at (3,-1.5) {$\overrightarrow{F}$};

    \end{scope}

\end{tikzpicture}
    \caption{The graph $\overrightarrow{F}$. The graphs $\overrightarrow{F_1}$ and $\overrightarrow{F_2}$ refer to the subgraphs of $F$ induced by the four top and four bottom vertices, respectively.
    \label{fig F}}
\end{figure}
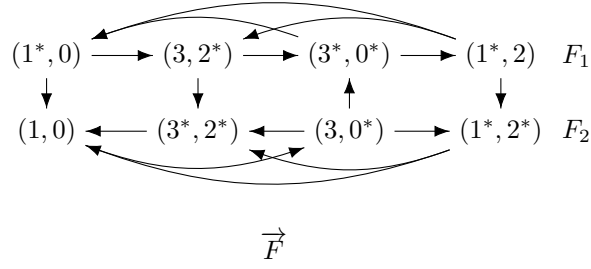

First we observe that their are $4$ choices for $X$. In the table below, we list those four choices along with the topological order they induce on the tournament $\overrightarrow{F}'[F_1]$. 
\begin{center}
        \begin{tabular}[h]{ll}
                 $X$ &  topological order \\\hline
                $(3^*,0^*)$        & $(1^*,2), (1^*,0), (3^*,0^*),( 3,2^*)$ \\ 
               $(3^*,0^*), (1^*,2)$  &        ($1^*,0), (3^*,0^*), (3,2^*), (1^*,2)$ \\
                $(3^*,0^*), (3,2^*)$  &        $(3,2^*), (1^*,2), (1^*,0), (3^*,0^*)$ \\
                $(3^*,0^*), (1^*,2), (1^*,0)$    &        $(3^*,0^*), (3,2^*), (1^*,2), (1^*,0)$ 
        \end{tabular}
\end{center}

Similarly, there are $8$ choices for $Y$ that are listed below: 

       $$ \{ (3,0^*) \}, \{ (1,0) \}, \{ (3,0^*), (1^*,2^*)\}, \{ (1,0), (3,0^*) \}, \{(3^*,2^*), (1^*,2^*)\}, \{(1,0), (3^*,2^*)\},$$ 
       $$\{ (3,0^*), (3^*,2^*), (1^*,2^*)\}, \{(1,0), (3^*,2^*), (1^*,2^*)\} . $$ 

\medskip
  
Therefore, let us do the case analysis based on the choices of $X$. 
Let $\overrightarrow{F}''$ denote the directed graph obtained from $\overrightarrow{F}$ by pushing the vertices of $X$. The explicit drawing of $\overrightarrow{F}''$ is provided below for each of the $4$ choices of $X$. 
It is tedious, but easy to observe that, for each choice of $X$, there does not exist any choice of $Y$ (the options are the eight choices listed above) which will convert $\overrightarrow{F}''$ into a directed acyclic graph (after pushing the vertices of $Y$).

\begin{enumerate}[(i)]
    \item If $X = \{ (3^*,0^*) \}$, then $\overrightarrow{F}''$ is as follows:
        
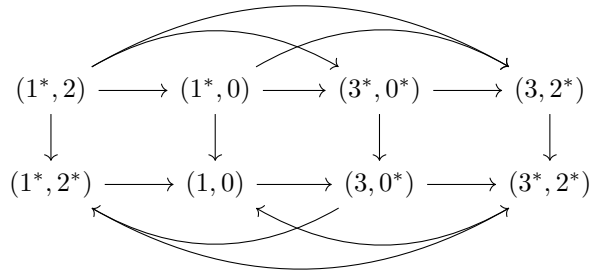
\begin{center}
            \begin{tikzcd}
                    (1^*,2)\arrow[rr,bend left]\arrow[rrr, bend left] \arrow[d]\arrow[r] & (1^*,0)\arrow[rr,bend left] \arrow[d]\arrow[r] &  (3^*,0^*) \arrow[d]\arrow[r] & (3,2^*) \arrow[d]\\
                    (1^*,2^*) \arrow[rrr,bend right]\arrow[r] & (1,0) \arrow[r] &  (3,0^*)\arrow[ll, bend left] \arrow[r] & (3^*,2^*)\arrow[ll, bend left]
            \end{tikzcd}
        \end{center}

    \item $X = \{ (3^*,0^*), (1^*,2) \}$, then $\overrightarrow{F}''$ is as follows:
                
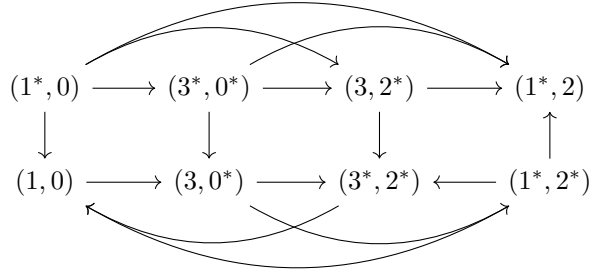
\begin{center}
                        \begin{tikzcd}
                                (1^*,0) \arrow[rr, bend left]\arrow[rrr, bend left]\arrow[d]\arrow[r] & (3^*,0^*) \arrow[rr, bend left]\arrow[d]\arrow[r] &  (3,2^*) \arrow[d]\arrow[r] & (1^*,2 )\arrow[d,<-]\\
                                (1,0) \arrow[r]\arrow[rr,bend right,<-]\arrow[rrr,bend right,<-] & (3,0^*) \arrow[rr,bend right]\arrow[r] &  (3^*,2^*) \arrow[r,<-] & (1^*,2^*)
                        \end{tikzcd}
                \end{center}

        \item $X = \{ (3^*,0^*), (3,2^*) \}$, then $\overrightarrow{F}''$ is as follows:
                
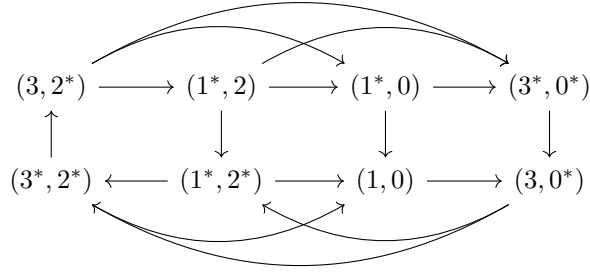
\begin{center}
                        \begin{tikzcd}
                                (3,2^*) \arrow[rr, bend left]\arrow[rrr, bend left]\arrow[d,<-]\arrow[r] & (1^*,2) \arrow[rr, bend left]\arrow[d]\arrow[r] &  (1^*,0) \arrow[d]\arrow[r] & (3^*,0^*) \arrow[d]\\
                                (3^*,2^*)\arrow[rr,bend right]\arrow[rrr,bend right,<-] \arrow[r,<-] & (1^*,2^*) \arrow[r]\arrow[rr,bend right,<-] &  (1,0) \arrow[r] &        (3,0^*)
                        \end{tikzcd}
                \end{center}

        \item $X = \{(3^*,0^*), (1^*,2), (1^*,0)\}$, then $\overrightarrow{F}''$ is as follows:
                
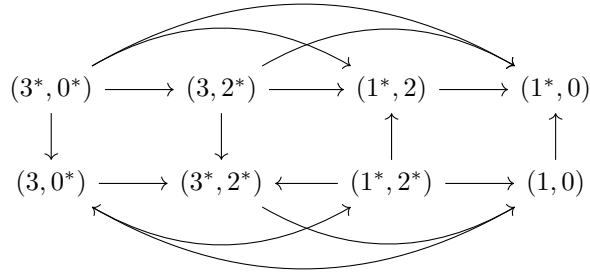
\begin{center}
                        \begin{tikzcd}
                                (3^*,0^*) \arrow[rr, bend left]\arrow[rrr, bend left]\arrow[d]\arrow[r] & (3,2^*) \arrow[rr, bend left]\arrow[d]\arrow[r] &  (1^*,2) \arrow[d,<-]\arrow[r] & (1^*,0) \arrow[d,<-]\\
                                (3,0^*) \arrow[r]\arrow[rr,bend right]\arrow[rrr,bend right,<-] & (3^*,2^*) \arrow[r,<-]\arrow[rr,bend right] &  (1^*,2^*) \arrow[r] & (1,0)
                        \end{tikzcd}
                \end{center}
\end{enumerate} 

That shows it is not possible to push some vertices of $\overrightarrow{F}$ and convert it into a directed acyclic graph $\overrightarrow{F}'$, and thus it concludes the proof. 
\end{proof}

\section{Non-multiplicative families of directed graphs}\label{sec directed non-multiplicative}
Furthermore, we explore a connection between 
pushably non-multiplicative directed graphs 
and non-multiplicative directed graphs.

\begin{lemma}\label{lem non-multiplicative}
    If $\overrightarrow{K}$ is a pushably non-multiplicative directed graph, then $AT(\overrightarrow{K})$ is a non-multiplicative directed graph. 
\end{lemma}

\begin{proof}
    Suppose $\overrightarrow{K}$ is a pushably non-multiplicative directed graph. 
    That means, there exists $\overrightarrow{G}$
    and $\overrightarrow{H}$ 
    such that 
    $\overrightarrow{G} \times_p \overrightarrow{H} \xrightarrow{push} \overrightarrow{K}$
    but 
    $\overrightarrow{G} \centernot{\xrightarrow{push}} \overrightarrow{K}$ and 
    $\overrightarrow{H} \centernot{\xrightarrow{push}} \overrightarrow{K}$.
    Since $AT(\overrightarrow{G}) \times AT(\overrightarrow{H}) = 
    AT(\overrightarrow{G} \times_p \overrightarrow{H})$, 
    by Theorem~\ref{thm push-product}, we can conclude that 
    $AT(\overrightarrow{G}) \times AT(\overrightarrow{H}) \rightarrow AT(\overrightarrow{K})$, but 
    $AT(\overrightarrow{G}) \not\rightarrow AT(\overrightarrow{K})$
    and
    $AT(\overrightarrow{H}) \not\rightarrow AT(\overrightarrow{K})$. 
    That means, $AT(\overrightarrow{K})$ is non-multiplicative. 
\end{proof}

Based on the above theorem, we are able to 
find new (infinite) families of non-
multiplicative directed graphs contributing 
in the answer of Question~\ref{ques digraph 
multiplicative}, and also building new tools 
to find non-multiplicative directed graphs.

\begin{theorem}\label{th non-multiplicative directed}
    The directed graph $AT(\overrightarrow{G})$ is non-multiplicative if $\overrightarrow{G}$
    is  an unbalanced bipartite directed graph, or an unbalanced oriented cycle, or a transitive tournament on at least three vertices. 
\end{theorem}

\begin{proof}
 The proof directly follows from Theorems~\ref{thm push-oriented-bipartite}, \ref{thm push-oriented-cycle}, and \ref{thm push-oriented-TT} using Lemma~\ref{lem non-multiplicative}.     
\end{proof}

\section{Conclusions}\label{sec conclusions}
\noindent 
(1) We present the state of the art 
of multiplicativity research in \textbf{Und}, \textbf{Dir}, \textbf{Push} in a tabular form (see Table~\ref{tab result}).

\begin{table}[h]
\centering
\setlength{\tabcolsep}{10pt}        
\renewcommand{\arraystretch}{1.5}   
\begin{adjustbox}{width=\columnwidth,center}

\begin{tabular}{|c|p{5cm}|p{5cm}|}
\hline
\textbf{Category} & \textbf{Multiplicative} & \textbf{Non-multiplicative} \\
\hline

\multirow{4}{*}{Undirected Graphs}
& $K_1, K_2, K_3$~\cite{el-zahar-sauer}
& $K_n$ for $n \ge 4$~\cite{tardiff-kn4,shitov-counter} \\
\cdashline{2-3}
& Cycles $C_n$~\cite{hell_oriented_multiplicativity}
& \multirow{3}{*}{} \\
\cdashline{2-2}
& Circular cliques $K_{{n/k}}$ with ${n/k} < 4$~\cite{tardiff-circ-cliq} &
\\
\cdashline{2-2}
& Square-free graphs~\cite{wrochna-square-free} &
\\
\hline

\multirow{4}{*}{Directed Graphs}
& Oriented paths and cycles homomorphically equivalent to directed paths~\cite{nevsetvril-pultr}
& 
All oriented cycles not listed among the multiplicative directed graphs. 
\\
\cdashline{2-3}
& $\mathcal{C}$-cycles~\cite{hell_oriented_multiplicativity}
& \cellcolor{green!25} Anti-twinned directed graphs of unbalanced bipartite graphs, odd oriented cycles, transitive tournaments on 3 or more vertices [Theorem~\ref{th non-multiplicative directed}]\\
\cdashline{2-2}
& Directed cycles of prime power length~\cite{haggkvist_multiplicative_1988}
&  \\
\cdashline{2-2}
& Transitive tournaments~\cite{nevsetvril-pultr}
&  \\

\hline

\multirow{4}{*}{Pushable digraphs}
& \cellcolor{green!25} Balanced bipartite graphs
  \newline [Theorem~\ref{thm push-oriented-bipartite}]
& \cellcolor{green!25} Unbalanced bipartite graphs
  \newline [Theorem~\ref{thm push-oriented-bipartite}] \\
\cdashline{2-3}
&
&
\cellcolor{green!25} All odd oriented cycles
  \newline [Theorem~\ref{thm push-oriented-cycle}] \\
\cdashline{3-3}
&
&
\cellcolor{green!25} Transitive tournaments on $n\ge 3$
  \newline [Theorem~\ref{thm push-oriented-TT}] \\
\hline

\end{tabular}
\end{adjustbox}
\vspace{0.5cm}
\caption{Summary of known characterization of multiplicative (directed) graphs related results  in the categories of undirected graphs, directed graphs, and pushable directed graphs. The results proved in this article are highlighted using green color.}
\label{tab result}

\end{table}

\medskip

\noindent
(2) 
We have established the existence of exponential objects in the category of pushable digraphs in Theorem~\ref{thm push-exponential} resolving Question~\ref{ques push exponential} from~\cite{das2026update}. We managed to use the exponential objects and their properties as a tool to prove results related to pushable multiplicative directed graphs (in Theorems~\ref{thm push-oriented-cycle} and~\ref{thm push-oriented-TT}, to be specific). 
We wonder if it is possible to 
prove an analogue of the `Denstiy Theorem' known 
for undirected and directed graphs (with respect to ordinary homomorphism), for directed graphs with respect to the pushable homomorphism.

\medskip

\noindent
(3) 
Our results partially address Question~\ref{ques pushably digraph multiplicative} which asks to characterize all pushably multiplicative directed graphs. However, the complete resolution of the question is far from complete. We propose the following ambitious conjecture, which proposes a complete solution to Question~\ref{ques pushably digraph multiplicative}. 

\begin{conjecture}
A directed graph $\overrightarrow{G}$ is pushably multiplicative if and only if it is balanced.
\end{conjecture}

\medskip

\noindent
(4) 
As a natural intermediate step towards the above conjecture would be to determine the multiplicativity of tournaments with respect to  pushable homomorphisms. 
In particular, one may try to prove that every tournament on at least $n$ vertices is pushably non-multiplicative, where $n \geq 3$. 
Note that this claim is correct for $n = 3$ since the directed cycle on $3$ vertices is pushably non-multiplicative according to Theorem~\ref{thm push-oriented-cycle}. 
In fact, we have proofs
of this claim for $n =4$ and $n=5$ cases (not included in this article). Thus, it will not be surprising to find a proof of the claim for all $n \geq 4$. 

\medskip

\noindent
(5) The techniques developed here rely heavily on the exponential directed graphs.  
It would be worthwhile to examine whether similar ideas can be adapted to the classical categories \textbf{Dir} and \textbf{Und}, possibly yielding new insights into multiplicativity or related structural phenomena.

\medskip

\noindent
(6) 
Finally, an intriguing direction is the study of multiplicativity in other graph categories where categorical product exists, such as, the category of signed graphs. 
We did explore this direction a bit, and surprisingly found that most of the signed graphs are non-multiplicative. We are yet to resolve the analogue of Question~\ref{ques multiplicative} for signed graphs though. We will probably record our findings for signed graphs in a different article.

\bibliographystyle{abbrv}
\bibliography{references}

\end{document}